\documentclass[a4paper,hidelinks, 12pt]{article}
\usepackage[utf8]{inputenc}
\usepackage{a4wide}
\usepackage{amsmath}
\usepackage{amsfonts}
\usepackage{amssymb}
\usepackage{amsthm}
\usepackage{mathabx}
\usepackage{array}

\usepackage{subfigure}
\usepackage{cite}
\usepackage{xcolor}
\usepackage{soul}
\usepackage{cancel}
\usepackage[toc,page]{appendix}
\numberwithin{equation}{section}
\usepackage{mathtools}
\usepackage{hyperref}
\usepackage{url}
\usepackage{enumitem}
\usepackage{algorithm}
\usepackage{float}
\usepackage{lipsum}
\usepackage{thmtools}

\allowdisplaybreaks

\newtheorem{theorem}{Theorem}
\newtheorem{corollary}{Corollary}

\theoremstyle{plain}

\newcommand{\La}{\mathfrak}
\newcommand{\Lg}{\mathsf}

\newcommand{\N}[1]{\mathbb{#1}}
\newcommand{\nn}{\nonumber}
\newcommand{\mc}{\mathcal}
\newcommand{\on}{\operatorname}
\newcommand{\pr}{\operatorname{pr}}
\begin{document}
\begin{titlepage}
\begin{center}

{\large \bf {Lie symmetry analysis of the two-Higgs-doublet model field equations}}

\vskip 1cm

M. Aa. Solberg\footnote{E-mail: marius.solberg@ntnu.no} 

\vspace{1.0cm}

Department of Structural Engineering, \\ Norwegian University of Science and Technology, \\
Trondheim, Norway\\

\end{center}

\vskip 3cm

\begin{abstract}
    We apply Lie symmetry analysis of partial differential equations (PDEs) to the Euler-Lagrange equations
		of the two-Higgs-doublet model (2HDM), to determine its scalar Lie point symmetries. A Lie point symmetry is a structure-preserving transformation of the spacetime variables and the fields of the model, which is also continuous and connected to the identity. Symmetries of PDEs may, in general, be divided into strict variational, divergence and non-variational symmetries, where the first two are collectively referred to as variational symmetries. Variational symmetries are usually 
	preserved under quantization, and variational Lie symmetries yield conservation laws. We demonstrate that there are no scalar Lie point divergence symmetries or non-variational Lie point symmetries in the 2HDM, and re-derive its well-known strict variational Lie point symmetries, thus confirming the consistency of our implementation of Lie's method.
  	Moreover, we prove three general results that may simplify Lie symmetry calculations for a wide class of particle physics models. Lie symmetry analysis of PDEs is broadly applicable 
	for determining Lie symmetries. As demonstrated in this work, the method can be applied to models with many variables, parameters, and reparametrization freedom, while any missing discrete symmetries can be identified through the automorphism groups of the resulting Lie symmetry algebras.
\end{abstract}

\end{titlepage}

\setcounter{footnote}{0}

\tableofcontents

\section{Introduction}
\label{sect:intro}
The discovery of the Standard Model (SM) Higgs boson at the LHC in 2012 by the
ATLAS \cite{ATLAS:2012yve} and CMS \cite{CMS:2012qbp} collaborations has
motivated a broad exploration of theories with extended scalar sectors. Such extensions may help explain the observed matter-antimatter asymmetry of
the universe, starting from Sakharov's conditions for baryogenesis
\cite{Sakharov:1967dj} and their possible realization in electroweak
baryogenesis scenarios. In particular, multi-Higgs-doublet models such as two-Higgs-doublet models
(2HDMs) \cite{Turok:1990zg} and three-Higgs-doublet models (3HDMs)
\cite{Weinberg:1976hu} can provide additional sources of $\on{CP}$ violation
and may strengthen the electroweak phase transition.

 Extended scalar sectors can also provide viable dark matter candidates. Gauge
singlet scalars offer a minimal realization \cite{McDonald:1993ex}, whereas
inert-doublet models provide an alternative framework \cite{Barbieri_2006}.
The inert-doublet scenario has been analysed in detail in \cite{Honorez_2007},
while associated gamma-ray line signatures were investigated in
\cite{Gustafsson_2007}.
 Additional
realizations arise in mixed doublet--singlet models, such as singlet-extended
2HDMs \cite{Drozd:2014yla}, and in multi-doublet constructions with
stabilizing symmetries, e.g.\ $\Lg{U}(1)$-stabilized 3HDMs \cite{Kun_inas_2024}.
As indicated, phenomenological studies have largely focused on additional
complex $\Lg{SU}(2)_L$ doublets \cite{ROSS1975135} and on scalar gauge
singlets. A comprehensive textbook treatment is given in \cite{Gunion:1989we},
while a more recent review can be found in \cite{Ivanov_2017}.

The introduction of such fields greatly enlarges the number of viable
particle physics models. One of the most important characteristics of a model
is its set of symmetries. By a symmetry we mean a structure-preserving
transformation of objects such as field equations, Lagrangians, and action
integrals.
Symmetries may imply conservation laws, relate apparently distinct models, reduce the number of free parameters, protect parameters from large quantum corrections, ensure that parameter relations are stable under quantum corrections and stabilize 
dark matter candidates against decay.    

Symmetry properties have been studied extensively for multi-Higgs-doublet
models (NHDMs), with the 2HDM as the simplest and most widely analysed
example. In particular, the 2HDM can accommodate $\on{CP}$ violation in the
scalar sector \cite{Lee:1973iz}, making it a natural framework for scenarios
of electroweak baryogenesis. It also plays a central role in dark matter
model building and in supersymmetric extensions of the SM. A two-Higgs-doublet
structure is, for example, an essential ingredient of minimal supersymmetric
models; see \cite{Fayet:1975yh} and the subsequent discussion in
\cite{FLORES198395}, as well as the detailed analysis in
\cite{GUNION19861}.

Continuous symmetries connected to the identity, also known as \textit{Lie symmetries}, may be studied in general by a method originally introduced by Norwegian mathematician Sophus Lie \cite{Lie1893}, and subsequently developed by his successors.
The method, Lie symmetry analysis of partial differential equations (PDEs), can be utilized to find all infinitesimal (sometimes called "local") symmetries of a system of PDEs $\Delta$, which correspond to the Lie symmetry algebra $\La{g}$ of $\Delta$. This method is the original context in which the concepts now known as Lie algebras and Lie groups emerged \cite{Lie1893}. The Lie symmetry algebra $\La{g}$ of $\Delta$ determines, to a large extent, the maximal global symmetry group $\Lg{G}$ of $\Delta$, because the Lie algebra 
 of $\Lg{G}$ must equal $\La{g}$. Lie symmetry analysis  may be applied to almost any kind of system of PDEs, for example a system of Euler-Lagrange equations $E(\mc{L})=0$ for a theory given by a Lagrangian $\mc{L}$. When the dependent variables of $\mc{L}$ are fields, the Euler-Lagrange equations are also referred to as the \textit{field equations} of the field theory given by $\mc{L}$. Particle physics models, such as multi-Higgs models, are examples of field theories given by Lagrangians. 

The field equations follow from the Lagrangian $\mc{L}$ via a variational principle. Consider the action $\mc{S}$, 
\begin{align}
	\mc{S}=\int_\Omega \mc{L} d^4 x.
\end{align}  
Then, by varying the fields $y^i$ of $\mc{L}$ infinitesimally, with vanishing 
variations $\delta y^i$ on the boundary $\partial \Omega$ of $\Omega$, and by demanding stationary values of $\mc{S}$ to first order of the infinitesimal parameter, the result is the 
Euler-Lagrange equations. For a Lagrangian containing at most first-order derivatives (as in the 2HDM), the Euler-Lagrange equations read
\begin{align}\label{E:feqs4d}
	\frac{\partial \mc{L}}{\partial y^i}-d_\mu \frac{\partial \mc{L}}{\partial y^i_{,\mu}}=0,
\end{align}
   for each field (i.e.~dependent variable) $y^i$ and derivative $y^i_{,\mu}\equiv \partial y^i/\partial x^\mu$ occurring in $\mc{L}$, where $d_\mu\equiv d/dx^\mu$ is a total derivative. Then, there are three types of symmetries
		of the field equations \eqref{E:feqs4d}: First, the \textit{strict variational symmetries} (SVSs) that leave the action $\mc{S}$ strictly invariant, that is,
 \begin{align}
	\Delta \mc{S} \equiv \hat{\mc{S}}-\mc{S}\equiv 0, 	
\end{align}
	where $\hat{\mc{S}}$ is the action transformed by the symmetry.
		Second, the \textit{divergence symmetries}, where 
		\begin{align}\label{E:introDivSym}
		\Delta \mc{S}= \int_\Omega d_\mu \beta^\mu d^4x=\int_{\partial \Omega}\beta^\mu dF_\mu=0,  	
		\end{align}
		is the integral of a total divergence $d_\mu \beta^\mu$, with 
		non-vanishing, local functions $\beta^\mu$. 
	Here, $\Delta \mc{S}$ is, by the divergence theorem, converted to a boundary term, which vanishes under the usual boundary conditions on the fields, as shown in the last equality of \eqref{E:introDivSym}. These first two types of symmetries are the \textit{variational symmetries}, symmetries
		of the action $\mc{S}$. Variational symmetries are typically preserved in the corresponding quantum theory, provided both the action $\mc{S}$
 and the measure of Feynman’s path integral remain invariant. Moreover, variational Lie symmetries lead to conserved currents by Noether's theorem \cite{Noether1918}.  
		The third type of symmetries are the \textit{non-variational symmetries} that preserve the structure of the field equations, but do not leave the action $\mc{S}$ invariant. 
		
When a symmetry only transforms
free and dependent variables $(x,y)$, it is called a \textit{point symmetry}. The derivatives of the dependent variables are also transformed under a point symmetry, however, these transformations are dictated by the transformations of the variables $(x,y)$. This implies that a point symmetry cannot, for instance, interchange a dependent variable with its derivative.		

The possible, inequivalent (real) scalar point SVS groups of the 2HDM were determined in \cite{PhysRevD.77.015017}, using a formalism of gauge invariant scalar bilinears. 
A new discrete point transformation in the 2HDM that leads to parameter relations respected by quantum corrections, characteristic of variational symmetries, was identified in 
\cite{Ferreira:2023dke}. As suggested in \cite{Ferreira:2023dke}, this transformation may be regarded as a complex, discrete symmetry. Such transformations may be a generic feature of quantum field theories, although they do not necessarily need to be interpreted as complex symmetries \cite{Trautner:2025yxz}. However, complex symmetries of real structures are not unknown in the context of symmetry analysis of differential equations, see e.g.~\cite{hydon2000symmetry} or \cite{Gray2013AutomorphismsOL} for examples of complex, discrete symmetries of \textsl{real} differential equations.
An overview of the 2HDM can be found in \cite{Branco_2012}.    
The possible scalar point SVSs of the 3HDM have been classified in stages:
Abelian symmetry groups were analysed in \cite{Ivanov_2012}, and the full
classification was completed in \cite{Ivanov_2013}. The corresponding
categorization of scalar symmetry groups in the 4HDM has only recently been
initiated, with extensions of cyclic groups studied in \cite{Shao_2023}
and extensions by rephasing groups in \cite{Shao_2024}. In these works,
finite group theory plays a central role.

The scalar point SVSs of the 2HDM and 3HDM, including custodial symmetries,
have also been classified. For the 2HDM, this programme was initiated in
\cite{Battye:2011jj} and completed in \cite{Pilaftsis:2011ed}. For
the 3HDM, the corresponding classification was obtained in
\cite{Darvishi:2019dbh}, using bilinear and tensor-product methods. Finally, the scalar, Lie point SVSs of the general NHDM kinetic terms were determined in \cite{Olaussen:2010aq}, including the corresponding symmetries in the custodial limit.  

The purpose of this work is twofold: First, inspired by
the aforementioned discovery in \cite{Ferreira:2023dke} of a new, complex 2HDM symmetry,  
we want to investigate the possibility of having Lie point divergence and non-variational symmetries in a 2HDM. Because all Higgs doublets carry the same quantum numbers (i.e., the same isospin and hypercharge), unitary linear combinations of the original doublets, accompanied by appropriate redefinitions of parameters to keep the Lagrangian invariant, relate different parameterizations that describe the same physics. This phenomenon is known as basis freedom or \textit{reparametrization freedom}. In our analysis, we will choose specific bases where certain parameters vanish, in order to minimize the occurrence of equivalent symmetries across different bases.
Second, we would like to demonstrate how to classify 
all Lie point symmetries in models with many variables, parameters and reparametrization freedom, like the 2HDM. Both of these aims will be pursued by applying Lie symmetry analysis to the field equations of the general 2HDM.


\paragraph{Structure of the paper}
\label{sec:StructureOfThePaper}
Sections \ref{sec:PointSymmetriesOfSystemsOfPDEs}-\ref{sec:SymmetriesOfTheActionMcS} review the relevant theory of Lie symmetry analysis of PDEs, while we present three new results in Sections \ref{sec:MultiHiggsModels} and \ref{sec:TheScalarSymmetriesAreTheSameInAnySpacetimeDimension}, where the first result, Theorem \ref{P:prXannTannV}, is crucial for Section \ref{sec:ParameterCasesAndReductionsOfThePotential}. 
We then provide a simple, concrete example of how the theory of Sections \ref{sec:PointSymmetriesOfSystemsOfPDEs}-\ref{sec:TheScalarSymmetriesAreTheSameInAnySpacetimeDimension} may be applied in Section \ref{sec:ASimpleExamplePhi4Theory}. Section \ref{sec:ASimpleExamplePhi4Theory} is not necessary for
reading Section \ref{sec:2HDM}, and may be skipped by readers
familiar with Lie symmetry analysis of PDEs or readers who want to go straight to the main application in Section \ref{sec:2HDM}.
In Section 
  \ref{sec:2HDM}, we recall standard formalism of the 2HDM,
		and in Sections \ref{sec:FindingAndSolvingTheDeterminingEquations}-\ref{sec:ParameterCasesAndReductionsOfThePotential} we perform the Lie symmetry analysis of the 2HDM. A summary and outlook is given in Section \ref{sec:SummaryAndOutlook}, while a proof of a result of Section \ref{sec:TheScalarSymmetriesAreTheSameInAnySpacetimeDimension}, Proposition \ref{P:SimplifiedModel}, is delegated to Appendix \ref{sec:ProofOfProposition}.
\paragraph{Conventions and notation}
\label{sec:ConventionsAndNotation}
In this article we adopt the mathematicians' convention for Lie algebras, which means that a matrix Lie group $\Lg{G}$ with Lie algebra $\La{g}$
is generated by $\exp(\La{g})$, and not by 
$\exp(i\La{g})$ which would correspond to the physicists' convention. Moreover, 
$d_\mu\equiv d/dx^\mu$ denotes the total derivative, whereas 
$D_\mu$ is reserved for the covariant derivative \eqref{E:covarDerDef}, unless the index is a multi-index $J$, see Section \ref{sec:InfinitesimalGenerators}, or a function $P$, implying a Fréchet derivative $D_P$. Repeated indices are implicitly summed over,
while a check 
over the indices implies that the summation convention is dispensed so that, 
for instance, $\eta^{\widecheck{\imath}} \partial_{\phi_{\widecheck{\imath}}}$ only consists of one term. 
Finally, we will use
respectively $\Re$ and $\Im$ for real and imaginary parts of expressions, e.g.~will $\lambda_5=\Re(\lambda_5)+i\, \Im(\lambda_5)$.
\section{Lie symmetry analysis of PDEs}
\label{sec:LieSymmetryAnalysisOfPDEs}
In this section, we provide a brief overview of the Lie symmetry theory
for systems of PDEs. A classical introduction is given by Olver
\cite{olver1998applications}, with a more concise account in his lecture
notes \cite{OlverLecturesLGaDE}. Further treatments are available in the
textbooks by Hydon \cite{hydon2000symmetry}, Bluman and Kumei
\cite{bluman2010applications}, and Cantwell \cite{cantwell2002introduction}.

\subsection{Point symmetries of systems of PDEs}
\label{sec:PointSymmetriesOfSystemsOfPDEs}

Consider an $n$th-order system of PDEs given by
\begin{align}\label{E:nthOrdPDEs}
	\Delta_i(x,y,y^{(1)},\ldots, y^{(n)})=0, \quad \forall i\in \{1,\ldots,m\}, 
\end{align}
with $d$ independent variables $x=(x^0,\ldots, x^{d-1})$ (spacetime coordinates) and
$q$ dependent variables $y=(y^1,\ldots,y^q)$ (the fields), where
$y^{(k)}$ denotes derivatives of order $k$ of the dependent variables $y^j$ with respect to the independent variables $x^\mu$. We abbreviate \eqref{E:nthOrdPDEs} as $\Delta=0$ or just $\Delta$.

Then, a point symmetry $S$ of the system of PDEs \eqref{E:nthOrdPDEs}
is a diffeomorphism\footnote{Which means that $S$ and $S^{-1}$ exist and have at least continuous first-order derivatives, but are often taken to be smooth, that is, $C^\infty$. 
} on the space of variables (i.e.,~it is a point transformation), which maps solutions of \eqref{E:nthOrdPDEs}
to solutions. More precisely, this means
\begin{align}
	S: U\subset \N{R}^{d+q} \to \N{R}^{d+q},\quad S\big((x,y)\big)= (\hat{x},\hat{y})
\end{align}
for an open set $U$, such that the transformed system of PDEs,
\begin{align}\label{E:nthOrdPDEsTransf}
	\Delta_i(\hat{x},\hat{y},\hat{y}^{(1)},\ldots, \hat{y}^{(n)})=0, \quad \forall i\in \{1,\ldots,m\}, 
\end{align}
 holds whenever equation \eqref{E:nthOrdPDEs} holds. In \eqref{E:nthOrdPDEsTransf}, the action of $S$ is prolonged to the derivatives, such that they are mapped to the corresponding derivatives in the transformed variables, that is,
\begin{align}
	S(\frac{d^k {y}^i}{d{x}^{\mu_1}\cdots d{x}^{\mu_k}})=\frac{d^k \hat{y}^i}{d\hat{x}^{\mu_1}\cdots d\hat{x}^{\mu_k}},
\end{align}
to preserve the structure of the original system of PDEs \eqref{E:nthOrdPDEs}. 
The condition for a diffeomorphism $S$ to be a symmetry of the system \eqref{E:nthOrdPDEs}
may now be expressed in compact form as 
\begin{align}
	\Delta =0 \Rightarrow \hat{\Delta} =0,
\end{align}
where $\hat{\Delta}\equiv \Delta(\hat{z})$, with $z$ representing all
 the arguments of $\Delta$, including the derivatives.

The expressions $\Delta_i$ in \eqref{E:nthOrdPDEs} can also be regarded as
functions on the \emph{$n$th jet space} $J^n$, which in the present setting
can be identified with Euclidean space $\mathbb{R}^s$ with formal coordinates
$(x,y,\ldots,y^{(n)})$ corresponding to all independent and dependent variables
and their distinct derivatives up to order $n$. This implies
$s = d + q\binom{d+n}{n}$. Each PDE can then be viewed as a smooth map
\begin{align}
  \Delta_i : J^n \to \mathbb{R},
\end{align}
and the system \eqref{E:nthOrdPDEs} can be described by the solution submanifold
$\mathcal{M}_\Delta \subset J^n$,
\begin{align}\label{E:solManifold}
  \mathcal{M}_\Delta
  = \bigl\{ (x,y,\ldots,y^{(n)}) \in J^n \;\big|\;
     \Delta_i(x,y,\ldots,y^{(n)}) = 0 \ \text{for all } i  \bigr\},
\end{align}
which consists of all points $(x,y,\ldots,y^{(n)}) \in J^n$ that satisfy the system. 

\subsection{Infinitesimal generators and their prolongations}
\label{sec:InfinitesimalGenerators}
An infinitesimal generator of a point transformation 
is a vector field
\begin{align}\label{E:infGenPoint}
	X= \xi^\mu(x,y)\frac{\partial}{\partial x^\mu} + \eta^i(x,y)\frac{\partial}{\partial y^i}, 
\end{align}
where $\mu$ is implicitly summed from $0$ to $d-1$ and likewise $i$ is summed from $1$ to $q$. The infinitesimal generator of a point symmetry $S$ shows
how an infinitesimal version of $S$ acts on the variables $z=(x,y)$,
\begin{align}
	S(z)=\hat{z}=z+\epsilon X(z),
\end{align}
  for an infinitesimal $\epsilon$. 
	The $k$th prolongation of $X$, denoted $\operatorname{pr}^{(k)}X$, is the
vector field on the $k$th jet space $J^k$ obtained by extending the action of 
$X$ to the derivatives of the
dependent variables, namely
\begin{align}\label{E:kProlong}
	\on{pr}^{(k)}X= X+ \sum_{1\leq|J|\leq k} \eta_J^i\frac{\partial}{\partial y^i_J},   
\end{align}
 where the multi-index $J=(j_0,\ldots,j_{d-1})$ encodes the distinct derivatives,
with norm
\begin{align}\label{E:|J|}
|J|=j_0+\cdots+j_{d-1},
\end{align}
where the derivatives of \eqref{E:kProlong} are given by
\begin{align}\label{E:y^i_J}
	y^i_J\equiv \frac{\partial^{|J|} y^i}{(\partial x^0)^{j_0}\cdots(\partial x^{d-1})^{j_{d-1}}},
\end{align}
with coefficients
\begin{align}
	\eta_J^i=D_J(Q^i)+\xi^\mu \frac{\partial y^i_J}{\partial x^\mu}
\end{align}
where the iterated total derivative 
\begin{align}\label{E:itTotder}
	D_J= (\frac{d}{dx^0})^{j_0}\cdots  (\frac{d}{dx^{d-1}})^{j_{d-1}},
\end{align}
and the characteristic
\begin{align}\label{E:characteristicDef}
	Q^i=\eta^i -\xi^\mu \frac{\partial y^i}{\partial x^\mu}.
\end{align}
We define the (infinite) prolongation
\begin{align}
  \operatorname{pr} X \equiv \operatorname{pr}^{(\infty)} X
\end{align}
as the formal vector field obtained by extending \eqref{E:kProlong} to all
derivatives of arbitrary order, so that $\operatorname{pr} X$ is the extension
of $X$ to the infinite jet space $J^\infty$.
If $\xi^\mu=0$ for all $\mu$ the infinitesimal point transformation generator $X$ is said to be in \textit{evolutionary form}, and for a general $X$ the infinitesimal generator 
\begin{align}\label{E:evolRepr}
	X_Q=Q^i\partial_{y^i}
\end{align}
 is called its \textit{evolutionary representative}.
 In the case $Q$ includes derivatives, $X_Q$ will be a so-called generalized symmetry of the system of PDEs if $X$ is a symmetry thereof \cite{olver1998applications}, but we will only consider cases where $X_Q$ is a point symmetry due to $\xi=0$ in \eqref{E:characteristicDef}.
 For a point transformation given by $X$ already in evolutionary form ($\xi=0$), its prolongation simply becomes
\begin{align}
	\on{pr}X= \on{pr}X_Q =\sum_{i,J}(D_JQ^i) \frac{\partial}{\partial y^i_J}=  \sum_{i,J}(D_J\eta^i) \frac{\partial}{\partial y^i_J}.
\end{align} 
In case of a very simple generator $X=y \partial_y$ and only one independent variable $x$ the 1-prolongation $\on{pr}^{(1)}(X)=y \partial_y+ y' \partial_{y'}$, and it hence just extends the infinitesimal transformation $y\to (1+\epsilon X) y= y+\epsilon y$ to the first-order derivative, since $y'\to (1+\epsilon \on{pr}^{(1)}(X)) y'= y'+\epsilon y'$. If the independent variables also transform, that is, $\xi^\mu \ne 0$ for some $\mu$, the prolongations may become much more complicated, as testified by the formulas above.

Now, let $\Lg{G}$ be a connected Lie group acting (locally) on $\mathbb{R}^{d+q}$
with coordinates $(x,y)$, where $x \in \mathbb{R}^d$ are the independent
variables and $y \in \mathbb{R}^q$ the dependent variables, and let $\La{g}$
be its Lie algebra. Then $\Lg{G}$ is a symmetry group\footnote{Possibly a
local group, i.e.~only defined in a neighbourhood of the identity.} of a
fully regular\footnote{``Fully regular'' here means that the system $\Delta$
is locally solvable and has a non-vanishing Jacobian determinant. In practice,
most systems are fully regular \cite{OlverLecturesLGaDE}.} system of $m$
PDEs, written as $\Delta=0$, if and only if
\begin{align}\label{E:prDelta=0PDEsymCond}
	 (\on{pr}X(\Delta_i))|_{\Delta = 0}=0 \quad \forall i\in \{1,\ldots,m\},  
\end{align}
for all infinitesimal generators $X\in \La{g}$ \cite{OlverLecturesLGaDE}. 
The condition \eqref{E:prDelta=0PDEsymCond} is equivalent to requiring that,
for a system $\Delta$ of order $n$, the prolongation $\operatorname{pr}^{(n)}X$
(or $\operatorname{pr}X$) is a vector field tangent to the solution manifold
$\mathcal{M}_\Delta$ defined in \eqref{E:solManifold}, which can also be
taken as a geometric definition of a (Lie point) symmetry.

The elements of $\Lg{G}$ are (Lie point) symmetries of the system $\Delta$.
When no confusion can arise, we will also refer to a generator
$X \in \La{g}$ as a symmetry whenever it satisfies
\eqref{E:prDelta=0PDEsymCond}.
 Moreover, all Lie point symmetries
of a system of PDEs will be found by Lie symmetry analysis,
because \eqref{E:prDelta=0PDEsymCond} holds for all symmetry generators in any connected symmetry group. We will only consider point symmetries, implemented by only applying point transformation generators \eqref{E:infGenPoint} in \eqref{E:prDelta=0PDEsymCond}, but \eqref{E:prDelta=0PDEsymCond} may be generalized to higher-order symmetries \cite{OlverLecturesLGaDE}.
Of course, 
we only 
have to consider prolongations $\on{pr}^{(n)}(X)$ up to the order $n$ of $\Delta$, when applying \eqref{E:prDelta=0PDEsymCond}. 
 If we perform a Lie symmetry analysis of the system $\Delta=0$
we may for example find that the Lie symmetry algebra 
is $\La{g}=\La{so}(6)$. Then, 
a (global) symmetry group of the system
 $\Delta=0$ may be $\Lg{SO}(6)$. However, if 
$-I$ acts identically to $I$ on $\Delta$, the corresponding faithful symmetry group is the projective group $\Lg{PSO}(6)\cong \Lg{SO}(6)/\{\pm I\}$. Both groups have the same Lie algebra 
$\La{g}$. However, the maximal, faithfully (or effectively) acting symmetry group may be a larger group than the latter, such as $\Lg{PO}(6)=\Lg{O}(6)/\{\pm I\}$, also with Lie algebra $\La{g}$, but with two components, due to the presence of an additional discrete reflection symmetry. Generally, if the symmetry algebra is $\La{g}$, the maximal symmetry group of $\Delta$ may, a priori, be any group $\Lg{G}$ with Lie algebra $\La{g}$. The identity component of $\Lg{G}$ is then a quotient 
$\tilde{\Lg{G}}/N$ of the unique simply connected group $\tilde{\Lg{G}}$ with Lie algebra $\La{g}$, where $N$ is a normal subgroup
\cite{kunzinger2024lietransformationgroups}.
The identification of the maximal symmetry group $\Lg{G}$, including any
missing discrete symmetries, can be achieved by studying the automorphism
groups of the Lie symmetry algebras obtained through Lie symmetry analysis.
For ordinary differential equations, Hydon \cite{hydon1998discrete}
describes how discrete point symmetries can be constructed from the continuous
symmetry algebra, and extends this approach to discrete contact symmetries in
\cite{hydon1998find}, while in \cite{hydon2000construct} similar methods are
applied to partial differential equations. A systematic application of these
techniques to determine the full symmetry group $\Lg{G}$ for the models
considered here would, however, require a separate analysis and is beyond the
scope of the present work.

Condition \eqref{E:prDelta=0PDEsymCond} is known as \textit{the linearized symmetry condition}, and the resulting concrete equations are called the \textit{determining 
equations} of $\La{g}$ (or, less precisely, $\Lg{G}$). Combining equations \eqref{E:kProlong} and \eqref{E:prDelta=0PDEsymCond} shows that the set of determining equations
yields an extensive, over-determined system of linear PDEs in the coefficients $\xi^{\mu}(x,y)$ and $\eta^{i}(x,y)$. This system may almost always be explicitly solved, and hence we can determine the symmetry algebra $\La{g}$, i.e.~the infinitesimal symmetries. In this work we, most of the time, will apply the \texttt{Mathematica} package \texttt{SYM} \cite{dimas2005sym} to calculate determining equations.
   
\subsection{Symmetries of the action $\mc{S}$}
\label{sec:SymmetriesOfTheActionMcS}
Of particular interest are symmetries of the action integral
\begin{align}
	\mc{S}=\int_{\Omega} \mc{L}({x},{y},\ldots,{y}^{(n)})\,dx^0\cdots dx^{d-1},
\end{align}
since they, if they are of Lie type, generate Noether currents and are usually preserved in the quantized theory. 
As mentioned in Section \ref{sect:intro}, a transformation is called a {variational symmetry} of a theory, if the transformation leaves the action integral of the theory invariant,
\begin{align}
	\mc{S}=\hat{\mc{S}}\equiv \int_{\hat{\Omega}} \mc{L}(\hat{x},\hat{y},\ldots,\hat{y}^{(n)})\,d\hat{x}^0\cdots d\hat{x}^{d-1}.
\end{align}
A Lie-type point transformation $\exp(\epsilon X)$ for $\epsilon \in \N{R}$ will be a variational symmetry if and only if
\begin{align}\label{E:condVarSym}
	\on{pr}X(\mc{L})+\mc{L}\,{d_\mu \xi^\mu}={d_\mu \beta^\mu},
\end{align}
where $\beta^\mu$ is a local\footnote{The $\beta$'s will be functions on the jet space $J^n$, hence non-local expressions such as $\int y dx$ cannot occur in the $\beta$'s \cite{olver1998applications}.} function of $x$, $y$ and derivatives of $y$ up to some order, and where $d_\mu\equiv {d}/{dx^\mu}$ is a total derivative. If we can choose 
	$\beta^\mu=0$ 
for all $\mu$,
and hence 
\begin{align}\label{E:condStrictVarSym}
	\on{pr}X(\mc{L})+\mc{L}\,{d_\mu \xi^\mu}=0,
\end{align}
the Lie point symmetry will be a {strict variational symmetry}, and if we cannot choose $\beta^\mu=0$ 
for all $\mu$,
it will be a {divergence symmetry}. 
For an infinitesimal Lie point divergence symmetry the action, due to the divergence theorem, transforms as
\begin{align}
	\hat{\mc{S}}=\mc{S}+\epsilon \int_{\partial \Omega}\beta^\mu dF_\mu + \mc{O}(\epsilon^2),  
\end{align}
for a small parameter $\epsilon$, where $dF_\mu$ are the components of the differential outward normal vector of the boundary $\partial \Omega$ of $\Omega$. Then, $\mc{S}$ is invariant if and only if the boundary term vanishes,
\begin{align}\label{E:boundTerm}
	\int_{\partial \Omega}\beta^\mu dF_\mu=0,
\end{align}
which will be the case if the fields vanish sufficiently fast at infinity or obey certain periodic properties. Note that this boundary term does not depend on variations of the fields; therefore, it is not automatically zero, unlike variations at the boundary. Furthermore, it could be tempting to apply the criterion \eqref{E:condVarSym} directly to find all divergence symmetries, but it typically will lead us to very difficult, if not unsolvable, differential equations, and we are advised to proceed via the Euler-Lagrange equations \cite{sym10120744}.

A Lagrangian $\mc{L}'=\mc{L}+\epsilon d_\mu\beta^\mu$
will yield the same Euler-Lagrange equations as $\mc{L}$, since a total divergence (also denoted a null Lagrangian) will always produce trivial Euler-Lagrange equations. 
In this study, the system of PDEs $\Delta=0$ is a set of Euler-Lagrange equations of a Lagrangian $\mc{L}$. We define the Euler operator $E=(E_1,\ldots,E_q)$ with components given by a formal infinite series
\begin{align}\label{E:EulerOpComps}
	E_i=\sum_{|J|\geq 0} (-1)^{|J|} D_J \frac{\partial}{\partial y^i_J}
	=\frac{\partial}{\partial y^i}-d_\mu \frac{\partial}{\partial y^i_{,\mu}}+\ldots,
	\quad \; i \in \{1,\ldots,q\}
\end{align}
where $|J|$, $y^i_J$ and $D_J$ are given by \eqref{E:|J|}, \eqref{E:y^i_J} and \eqref{E:itTotder}, respectively.  
Then, the Euler-Lagrange equations of a Lagrangian $\mc{L}$ can be written as
\begin{align}\label{E:ELeqs}
	E(\mc{L})=0.
\end{align}
A direct calculation shows that each $E_i$ annihilates any total divergence $d_\mu\beta^\mu$, as previously implied.
 In fact, it can be shown that a function $f$ is a total divergence if and only if it is annihilated by the Euler operator \cite{olver1998applications},
\begin{align}\label{E:EtotDiv}
	E(f)=0.
\end{align}
Therefore, if 
\begin{align}\label{E:divSymmetry}
	\on{pr}X(\mc{L})+\mc{L}\,d_\mu \xi^\mu=f,
\end{align}
for a non-vanishing $f$, the infinitesimal generator $X$ generates a divergence symmetry if and only if \eqref{E:EtotDiv}
	holds and hence $f=d_\mu \beta^\mu$ for some local functions $\beta^\mu$.
	
	The strict variational Lie point symmetries, given by the symmetry algebra $\La{g}_\text{svar}$, are obviously variational. The general variational Lie point symmetries, including divergence symmetries, also generate a Lie symmetry algebra, which we denote $\La{g}_\text{var}$. Finally, all Lie point symmetries of the Euler-Lagrange equations of $\mc{L}$ will generate a Lie symmetry algebra $\La{g}_\text{EL}$, which will contain the two other algebras, because it can be shown that all variational symmetries are symmetries of the Euler-Lagrange equations of the theory (the converse does not hold in general). 
Hence, the following symmetry algebra inclusions always hold:
\begin{align}\label{E:symAlgIncl}
	\La{g}_\text{svar} \subseteq \La{g}_\text{var}\subseteq \La{g}_\text{EL}\,.
\end{align}
The symmetries of $\La{g}_\text{EL}$ that are not included in the subalgebra $\La{g}_\text{var}$ are {non-variational Lie point symmetries}. 
In section \ref{sec:2HDM} we will demonstrate that for the 2HDM,  
 the inclusions in \eqref{E:symAlgIncl} actually are equalities. 

According to Noether's theorem, variational Lie point symmetries generate conserved currents. 
 In the case of a first-order Lagrangian $\mc{L}({x},{y},{y}^{(1)})$ a symmetry given by \eqref{E:infGenPoint} induces a conservation law $d_\mu j^\mu=0$ with a conserved current
\begin{align}
	j^\mu =(\eta^i-\xi^\nu y_{,\nu}^i)\frac{\partial\mc{L}}{\partial y^i_{,\mu}}+\xi^\mu \mc{L}-\beta^\mu.
\end{align}
The quantized version of the theory defined by the classical action $\mc{S}$, is determined by Feynman's path integral $Z$, 
\begin{align}
Z \;=\; \int \mc{D}y \; e^{\,i \mc{S}[y]}.	
\end{align}
Variational symmetries leave the action invariant, and if the measure $\mc{D}y\equiv \prod_{i,x}d y^i(x)$ of $Z$
is also invariant under the transformation, the symmetry will also be a symmetry of the quantized theory.

\subsection{Theories with potentials}
\label{sec:MultiHiggsModels}
One of the goals of this work is to demonstrate how Lie symmetry analysis of PDEs can be applied to the Euler-Lagrange equations of particle physics models with potentials, such as multi-Higgs models.  
In this section we present 
three results which may simplify the symmetry analysis of such models.

We start by defining the \textit{Fréchet derivative} \cite{olver1998applications} of an $r$-tuple $P[y]\equiv P(x,y,\ldots,y^{(n)})$ of differential functions as the operator 
$D_P$ that maps any $q$-tuple of differential (smooth) functions $U$ to
\begin{align}
	D_P[U] \;=\; \left.\frac{d}{d\varepsilon}\,P\big[y+\varepsilon\,U\big]\right|_{\varepsilon=0}\,.
\end{align}
Then, the Fréchet derivative of an $r$-tuple of functions $P=(P_1,\ldots,P_r)$ is the $r\times q$ matrix with elements \cite{olver1998applications}
\begin{align}
	(D_P)_{ij}=\sum_{|J|\geq 0} \frac{\partial P_i}{\partial y_J^j}D_J.
\end{align}
The \textit{adjoint} $D_P^\ast$ of the Fréchet derivative is an operator with similar properties \cite{olver1998applications},
\begin{align}\label{E:FrechetAdj}
	(D_P)_{ij}^\ast= \sum_{|J|\geq 0}  ((-1)^{|J|} D_J)\cdot \frac{\partial P_j}{\partial y_J^i}\,,
\end{align}
where e.g.~$D_x \cdot u(x)=u_x+uD_x$ on the right-hand side of equation \eqref{E:FrechetAdj}. Let $X$ be an infinitesimal generator with coefficients
\begin{align}
\xi^\mu &=0, \quad \forall \mu \in\{0,\ldots,d-1 \} \label{E:xiPPT}\\
\eta^i &=\eta^i(y^1,\ldots,y^{q}),	\quad \forall i \in\{1,\ldots,q \} \label{E:etaPPT}
\end{align}
where the $\eta^i$ are polynomials in the fields (i.e.,~the dependent variables).
Then, $X$ has characteristic 
\begin{align}
	Q^i=\eta^i(y^1,\ldots,y^{q}), 
\end{align}
cf.~\eqref{E:characteristicDef}, and $X$ is in evolutionary form because of \eqref{E:xiPPT}.
The adjoint of the Fréchet derivative of $Q$ then becomes  
\begin{align}\label{E:adjFrechetScAff}
	(D_Q)_{ij}^\ast=  \frac{\partial \eta^j(y^1,\ldots,y^{q})}{\partial y^i}=\mc{J}_{ji}=\mc{J}^T_{ij},
\end{align}
where $\mc{J}$ is the Jacobian matrix of $\eta$,
since only the $|J|=0$ term contributes to \eqref{E:FrechetAdj}.
Derivatives such as $d_\mu$ and $\partial_{y^i_{,\mu}}$ do not commute.
However, the following, useful commutation formula holds \cite{olver1998applications}
\begin{align}\label{E:commFrechet}
	E(\on{pr}X_Q(\mc{L}))=\on{pr}X_Q(E(\mc{L}))+D_Q^\ast E(\mc{L}),
\end{align}
for a $q$-tuple of general, smooth functions $Q^i$, cf.~\eqref{E:evolRepr}.
Moreover, let 
\begin{align}
\mc{L}=T-V 	
\end{align}
be the Lagrangian density of a theory, where $T$ denotes the kinetic part (the sum of the kinetic terms), while $V(\varphi_1,\ldots,\varphi_m)$ is a potential, i.e.~a real polynomial in a subset of the dependent variables. Specifically,
let
\begin{align}\label{E:varphiSubsety}
	\varphi=\{\varphi_1,\ldots,\varphi_m\}\subset \{y^1,\ldots,y^q\}=y.
\end{align}
For convenience, we assume that the dependent variables are ordered such that
\begin{align}
	\varphi_i\equiv y^i, \quad \forall i\in \{1,2,\ldots, m\}.
\end{align}
An $n$th-order Lagrangian $\mathcal{L}$ may be viewed as a function on the jet
space $J^n$, $\mathcal{L}:J^n\to\N{R}$. Note that $\varphi$ may consist of any of the fields of $y$, although the intention is to let $V(\varphi)=V(\phi)$, that is, let $V$ be a scalar potential. Also note that the "kinetic" part $T$ here actually may consist of any terms complementary to $V$. 
 Moreover, we assume $E(L)=0$ does not imply any polynomial consequences of the form  
\begin{align}\label{E:polCons}
	p(\varphi)=0,
\end{align}
where $p$ is a non-trivial, multivariable polynomial in the same fields as the potential $V$.\footnote{In differential algebra language, $I(E(\mc{L}))\cap \N{R}[\varphi_1,\ldots,\varphi_m]=\{0\}$, where $I(E(\mc{L}))$ is the differential ideal generated by the Euler-Lagrange expressions $E(\mc{L})$ (it encodes all formal consequences of the Euler-Lagrange-equations), see for example \cite{Robertz2018}, and $\N{R}[\varphi_1,\ldots,\varphi_m]$ is the ring of all polynomials in the variables $\varphi$ with real coefficients.} This means that we cannot derive any relation of the form \eqref{E:polCons}, from the Euler-Lagrange equations.
We will call such a theory a \textit{polynomial potential theory}. Practically any multi-Higgs model is an example of a polynomial potential theory, including the 2HDM. One reason is that each Euler-Lagrange equation $E_i(\mc{L})=0$ in multi-Higgs models will include distinct, second-order derivatives of the field $y^i$ that do not appear in any other $E_j(\mc{L})=0$ for $j\ne i$, and hence cannot be eliminated to produce some polynomial consequence \eqref{E:polCons}.
Here, we assume the multivariable polynomial $V$ may include any terms of any degree in the fields $\varphi$, but in some cases we will regard potentials without linear terms (but a constant in $V$ is allowed). 
However, in these cases, the linear terms are not forbidden among the kinetic terms.
 We now prove the following, useful result regarding evolutionary symmetries of polynomial potential theories: 
\begin{theorem}\label{P:prXannTannV}
Let $\mc{L}=T-V$ be a polynomial potential theory and let the infinitesimal generator 
\[
X=\eta^i(y^1,\ldots,y^{q})\partial_{y^i}, 
\]
where each $\eta^i$ is a polynomial, be a symmetry of $E(\mc{L})=0$. Moreover, assume that either $V(\varphi_1,\ldots,\varphi_m)$ does not contain any linear terms $\alpha_i \varphi_i$, or that
 $\alpha_i\ne 0 \Rightarrow \eta^i(y^1,\ldots,y^{q})$ does not
contain a constant term.
Then,   
\begin{enumerate}[label=(\roman*)]
 \item If $\on{pr}X(T)=0$, the symmetry generated by $X$ is strictly variational.
  \item If $\on{pr}X(T)=d_\mu \beta^\mu$ for some non-vanishing current $\beta^\mu$, i.e.~a total divergence, 
        the symmetry generated by $X$ is a divergence symmetry.
\end{enumerate}
\end{theorem}
\begin{proof}
 The linearized symmetry condition \eqref{E:prDelta=0PDEsymCond}
applied on the Euler-Lagrange equations $E(\mc{L})=0$ yields
\begin{align}\label{E:prDelta=0PDEsymCond2}
	 \big(\on{pr}X(E_i(\mc{L}))\big)|_{E(\mc{L}) = 0}=0, \quad \forall i\in \{1,\ldots,q\},  
\end{align}
while \eqref{E:commFrechet} implies
\begin{align}\label{E:commFrechet2}
	\on{pr}X(E(\mc{L}))= E(\on{pr}X(\mc{L}))-D_Q^\ast E(\mc{L}),
\end{align}
since $\on{pr}X= \on{pr}X_Q$, because the former is already in evolutionary form. Assume 
\begin{align}
\on{pr}X(T)=d_\mu \beta^\mu, 	
\end{align}
for a possible vanishing current $\beta$. Substituting \eqref{E:adjFrechetScAff} into \eqref{E:commFrechet2} then yields, 
\begin{align}\label{E:commFrechet3}
	\on{pr}X(E(\mc{L}))= -E(\on{pr}X(V))-\mc{J}^T E(\mc{L}),
\end{align}
because $E(d_\mu \beta^\mu)=0$ cf.~\eqref{E:EtotDiv}, and
where $\mc{J}$ is the Jacobian matrix of $\eta$.
Then, if $E(\mc{L}) = 0$, the first and last terms in \eqref{E:commFrechet3} vanish, cf.~\eqref{E:prDelta=0PDEsymCond2}, that is, for any $i$
\begin{align}\label{E:derPrPot0}
	E_i(\on{pr}X(V))=\frac{\partial}{\partial y^i}\on{pr}X(V)=0,
\end{align}
because there are no derivatives in the potential. The polynomial equation \eqref{E:derPrPot0} will also hold when $E(\mc{L})\ne 0$, because $E(\mc{L})= 0$ does not imply any polynomial relations (i.e.,~consequences) between the fields of $V$, as $\mc{L}$ is a polynomial potential theory. But then $\on{pr}X(V)=C$ for constant $C$, and $C=0$ because all terms of 
\begin{align}
\on{pr}X(V)=X(V)=\eta^i\partial_{y^i}V	
\end{align}
are, if non-vanishing, at least linear in the fields, because if $\partial_{y^i}V$ includes a constant term for an $i$, then $\eta^i$ does not, per assumption. This means that the potential $V$ is annihilated by 
the prolongation of $X$,
\begin{align}
	\on{pr}X(V)=0, 
\end{align}
  and hence 
	\begin{align}
		\on{pr}X(\mc{L})=\on{pr}X(T)=d_\mu \beta^\mu.
	\end{align}
Thus, if 
 $\beta=0$ then $X$ generates a strict variational symmetry, cf.~\eqref{E:condStrictVarSym} with $\xi=0$.
Furthermore, if $\beta$ is non-vanishing, $X$ generates a divergence symmetry, cf.~\eqref{E:divSymmetry} with $\xi=0$.
\end{proof}
  We will apply Theorem \ref{P:prXannTannV} in Section \ref{sec:2HDM} to demonstrate that symmetries must be strictly variational, without having to consider all the numerous, different conditions on the parameters of the potential.

\subsection{Scalar, variational symmetries in any spacetime dimension}
\label{sec:TheScalarSymmetriesAreTheSameInAnySpacetimeDimension}
It turns out that the purely scalar, variational Lie point symmetries of a multi-Higgs model with any number of doublets and singlets are the same for any spacetime dimension $d$. This may imply computational advantages, because if we are only interested in scalar, variational symmetries we may reduce the number of spacetime variables and gauge fields, and hence 
reduce the computational cost of finding the determining equations of the scalar, variational symmetries.  

Let the most general NHDM+KS Lagrangian in $d$ spacetime dimensions be given by
\begin{align}\label{E:L_NHDM+KS}
	\mc{L}_d(x_0,\ldots,x_{d-1})&=\sum_{\mu=0}^{d-1} \left( \sum_{n=1}^N(D_\mu \Phi_n)^\dag D^\mu \Phi_n+ \sum_{m=1}^K\frac{1}{2}\partial_\mu s_m\partial^\mu s_m\right) -V(\Phi,s) +T_\text{GB}
	\end{align} 
	with $\Phi\equiv (\Phi_1,\ldots,\Phi_N)$ and $s\equiv (s_1,\ldots,s_K)$, and
	where the kinetic terms of the gauge bosons read
	\begin{align}\label{E:TGB}
T_\text{GB}	= -\sum_{\mu, \nu=0}^{d-1} \left( \tfrac{1}{4}W^a_{\mu\nu}W^{a\mu\nu}
  +\tfrac{1}{4}B_{\mu\nu}B^{\mu\nu}\right).		
	\end{align}
The covariant derivatives, gauge boson field strength tensors and Higgs doublets are given by:
\begin{align}
D_\mu &= \partial_\mu + i g \frac{\sigma^a}{2} W^a_\mu + i g' Y B_\mu, \label{E:covarDerDef} \\
W^a_{\mu\nu} &= \partial_\mu W^a_\nu - \partial_\nu W^a_\mu + g \epsilon^{abc}W^b_\mu W^c_\nu, \label{E:Wfst}\\
B_{\mu\nu} &= \partial_\mu B_\nu - \partial_\nu B_\mu, \label{E:Bfst} \\
\Phi_j &= \frac{1}{\sqrt{2}}
\begin{pmatrix}
\phi_{4(j-1)+1} + i \phi_{4(j-1)+2} \\
\phi_{4(j-1)+3} + i \phi_{4(j-1)+4}
\end{pmatrix}, \label{E:HiggsDoubletsNHDM}
\end{align}
 where $\sigma^a$ are the Pauli matrices, $1\leq a \leq 3$ and $Y$ denotes the
$\Lg{U}(1)_Y$ hypercharge of the corresponding field. Geometrically,
$W^a_\mu$ can be viewed as the local components of an $\Lg{SU}(2)_L$
principal connection, while $B_\mu$ are the local components of a
$\Lg{U}(1)_Y$ principal connection.
Here, we also include the case where $K=0$,
that is, where the theory is a pure NHDM with no gauge singlets.

Then, the variational, scalar symmetries of theories given by $\mc{L}_d$ are independent of $d$: 
\begin{restatable}{proposition}{MainProp}\label{P:SimplifiedModel}
  The variational Lie point symmetries transforming only the scalars are the same for all NHDM+KS Lagrangians $\mc{L}_d$, regardless of the spacetime dimension $d\in \N{N}$.    
\end{restatable}
\begin{proof}
The proof of this is provided in Appendix \ref{sec:ProofOfProposition}.
\end{proof}
An immediate consequence of Proposition \ref{P:SimplifiedModel}
is then
\begin{corollary}\label{C:1d4dcorr}
To find all scalar, variational Lie point symmetries of an NHDM+KS Lagrangian $\mc{L}_4$, it is sufficient to consider the simplified Lagrangian $\mc{L}_1$, because such symmetries are exactly the same for the two theories.
\end{corollary}
The simplified Lagrangian $\mc{L}_1$ will only contain four gauge fields and one free variable and is hence computationally much less demanding than $\mc{L}_4$ that contains 16 gauge fields and four free variables.
The validity of Proposition \ref{P:SimplifiedModel} is confirmed by performing the analysis of Section \ref{sec:2HDM} on the $d=2$ variant 
of the 2HDM Lagrangian $\mc{L}_2$, in addition to the $d=4$
2HDM Lagrangian $\mc{L}_4$.\footnote{The even simpler case $d=1$ causes some \texttt{SYM}-related technical issues, because the kinetic terms of the gauge bosons vanish in this case. However, we could have added arbitrary dummy derivatives to circumvent this, as the kinetic gauge terms are annihilated by scalar symmetries $\on{pr}X$, and therefore become irrelevant for determining scalar symmetries.} After specializing to scalar transformations, the two approaches yield exactly the same equations, and hence exactly the same symmetry algebras (all Lie symmetries of the 2HDM turn out to be strictly variational and will hence be shared by the two Lagrangians according to Proposition \ref{P:SimplifiedModel}). However, in the $d=2$ case, the number of free variables and gauge fields are halved, that is, the field equations
\begin{align}
E(\mc{L}_2)=0	
\end{align}
form a system of 16 equations with 16 dependent and two free variables, whereas the field equations
\begin{align}
	E(\mc{L}_4)=0
\end{align}
form a system of 24 equations with 24 dependent and four free variables. The computation time for \texttt{SYM} to calculate the determining equations was reduced to a small fraction in the $d=2$ case.\footnote{For the 2HDM the calculation of the determining equations of $\mc{L}_4$ took nearly 14 hours on a desktop computer, whereas the corresponding calculation for $\mc{L}_2$ on the same computer took only 45 minutes.} 

In Appendix \ref{sec:ProofOfProposition}, we argue that Proposition \ref{P:SimplifiedModel} and Corollary \ref{C:1d4dcorr} cannot be extended to scalar, non-variational symmetries.

\section{An illustrative example: real, scalar $\phi^4$ theory}
\label{sec:ASimpleExamplePhi4Theory}
The following section may be skipped by readers familiar
with Lie symmetry analysis of PDEs. The aim is here to illustrate the general theory of Sections \ref{sec:PointSymmetriesOfSystemsOfPDEs}-\ref{sec:TheScalarSymmetriesAreTheSameInAnySpacetimeDimension}, using a simple example, namely
 real, scalar $\phi^4$ theory.
\subsection{$\phi^4$ theory in one dimension}
\label{sec:Phi4TheoryIn1dSpacetime}
Real scalar $\phi^4$ theory in 1d has a Lagrangian density
\begin{align}\label{E:LRscphi4}
	\mathcal{L} = \frac{1}{2}  (\phi')^2  - \frac{1}{2} m^2 \phi^2 - \frac{\lambda}{4!} \phi^4,
\end{align}
where $\phi=\phi(x)$ is a scalar field that depends on one free variable $x$, with a nonlinear field equation
\begin{align}\label{E:phi4_1D}
   \phi'' + m^2 \phi +\frac{\lambda}{6} \phi^3=0.
\end{align}
In the standard real $\phi^4$ model \eqref{E:LRscphi4} one usually imposes a discrete
$\mathbb{Z}_2$ symmetry $\phi \to -\phi$ in order to forbid linear and
cubic terms. This yields a simple, renormalizable model with a symmetric
vacuum structure and a potential containing only the even powers
$\phi^2$ and $\phi^4$.
The relevant prolongation \eqref{E:kProlong} of the infinitesimal symmetry generator $X$ in \eqref{E:infGenPoint} is 
\begin{align}\label{E:phi42prolong}
	\on{pr}^{(2)}X&=\xi(x,\phi) {\partial_x}  + \eta(x,\phi){\partial_\phi }   
	+ \big(\phi' \eta _{\phi }+\eta
   _{x}-\phi' \left(\phi' \xi _{\phi }+\xi _{x}\right)\big){\partial_{\phi'}} \nn \\
	&+ \Big(\phi'' \eta _{\phi }+\phi' \eta _{x\phi }+\phi'
   \left(\phi' \eta _{\phi \phi }+\eta _{x\phi }\right)+\eta
   _{xx}-2 \phi'' \left(\phi' \xi _{\phi }+\xi
   _{x}\right) \nn \\
	&-\phi' \left(\phi'' \xi _{\phi
   }+\phi' \xi _{x\phi }+\phi' \left(\phi' \xi _{\phi \phi
   }+\xi _{x\phi }\right)+\xi _{xx}\right)\Big){\partial_{\phi''}}\, ,
\end{align}
where subscripts indicate partial derivatives.
We then apply \eqref{E:phi42prolong} to \eqref{E:phi4_1D},
and thereafter use the substitution 
\begin{align}\label{E:substPhi4}
	\phi''\to - m^2 \phi -\frac{\lambda}{6} \phi^3
\end{align}
everywhere we can, cf.~the linearized
symmetry condition \eqref{E:prDelta=0PDEsymCond}, with the following determining equation as result:
\begin{align}\label{E:detEqsPhi4}
0&=\left(\tfrac{\lambda  \phi ^2}{2}+m^2\right) \eta -\tfrac{1}{6}
   \left(\lambda  \phi ^3+6 m^2 \phi \right) \left(\eta _{\phi }-3
   \phi ' \xi _{\phi }-2 \xi _{x}\right)\nn \\&-\phi ' \left(\phi '
   \left(-\eta _{\phi \phi }+\phi ' \xi _{\phi \phi }+2 \xi
   _{x\phi }\right)-2 \eta _{x\phi }+\xi _{xx}\right)+\eta
   _{xx} \, .
	\end{align}
 Here, each coefficient of each distinct power of $\phi'$
must vanish for \eqref{E:detEqsPhi4} to hold for any solution of the original 
field equation \eqref{E:phi4_1D} (remember, $\xi$ and $\eta$ do not contain $\phi'$), and we then obtain a larger system of determining equations:
\begin{align}\label{E:detEqsPhi4end}
{\phi'}^0:\quad	0&=\left(\tfrac{\lambda  \phi ^2}{2}+m^2\right) \eta -\tfrac{1}{6}
   \left(\lambda  \phi ^3+6 m^2 \phi \right) \left(\eta _{\phi }-2 \xi _{x}\right)\nn +\eta_{xx}\, , \nn \\
	{\phi'}^1:\quad	0&=\tfrac{1}{2} \xi_{\phi } \left(\lambda  \phi ^3+6 m^2 \phi \right)+2 \eta
   _{x\phi }-\xi _{xx}\, , \nn \\
	{\phi'}^2:\quad	0&= \eta _{\phi \phi }-2 \xi _{x\phi }\, , \nn \\
	{\phi'}^3:\quad	0&= \xi _{\phi \phi }\, .
\end{align}
The determining equations \eqref{E:detEqsPhi4end} are a second-order system
of homogeneous, linear PDEs in the coefficients of the infinitesimal generator $X$, that is, $\xi$
and $\eta$. The system can be solved by elementary methods, and the solutions will depend on whether the parameters $m^2$ and $\lambda$ vanish, as the vanishing parameters will yield simpler equations with more solutions.
\subsubsection{Massive $\phi^4$ theory}
\label{sec:MassivePhi4Theory}
Assuming 
\begin{align}
	m^2,\lambda\ne 0
\end{align}
the solutions of the set of four determining equations \eqref{E:detEqsPhi4end}
are 
\begin{align}
	\xi(x,\phi)=c_1, \quad
	\eta(x,\phi)=0,
\end{align}
for a constant $c_1$, which gives us the 1d Lie algebra defined by the generator
\begin{align}\label{E:LieAlg_phi^4}
	X_1= \partial_x,
\end{align}
 which generates the translational symmetry $x \to x + \epsilon$ and
abstractly spans the one-dimensional Lie algebra $\mathbb{R}$.
 The generator \eqref{E:LieAlg_phi^4} is here a strict variational symmetry, because 
	\begin{align}
	\on{pr}\partial_x (\mc{L})+\mc{L}\,\frac{d\cdot 1}{dx} = \partial_x (\mc{L})=0,
	\end{align}
 cf.~\eqref{E:condStrictVarSym}, since $\mc{L}$ does not contain $x$ explicitly.
	\subsubsection{Massless $\phi^4$ theory}
\label{sec:MasslessPhi4Theory}
	If 
	\begin{align}
		m=0 \quad \text{and} \quad \lambda\ne 0
	\end{align}
	in \eqref{E:phi4_1D} and \eqref{E:detEqsPhi4end}, 
	the solutions to \eqref{E:detEqsPhi4end} are
	\begin{align}
		\xi(x,\phi)=c_1+c_2 x,\quad \eta(x,\phi)= -c_2\phi, 
	\end{align}
	for arbitrary, real constants $c_1$ and $c_2$. 
	Hence, the Lie symmetry \eqref{E:LieAlg_phi^4} is enhanced to a 2d Lie algebra
\begin{align}
	 X_1 &=\partial_x , \nonumber  \\
 X_2 &= x\partial_x - \phi\partial_\phi, 
\end{align}
where the commutator equals
\begin{align}
	[X_1,X_2]=X_1,
\end{align}
and we therefore have the 2d, non-abelian Lie algebra $\La{a}(1)$. The generator $X_1$ is strictly variational in the same manner as before. Concerning $X_2$,
\begin{align}\label{E:prX21d}
	\on{pr}^{(1)}X_2(\mc{L})+ \mc{L}d_x \xi  =(x\partial_x-\phi \partial_\phi
	-2\phi'\partial_{\phi'})\mc{L}+\mc{L}\,\frac{d x}{dx}= -4\mc{L}+\mc{L}= -3\mc{L},
\end{align}
where $E(-3\mc{L})=-3E(\mc{L})$ is non-zero and is proportional to the left-hand side of the field equation. Hence, the symmetry is non-variational, cf.~\eqref{E:EtotDiv}, which means that it is a symmetry of the field equation, but not the action. That $X_2$ is actually a symmetry of the field equation is 
confirmed by
checking the linearized symmetry condition \eqref{E:prDelta=0PDEsymCond} for the field equation:
\begin{align}
	\on{pr}^{(2)}X_2 (E(\mc{L}))=(x\partial_x-\phi \partial_\phi
	-2\phi'\partial_{\phi'}-3\phi''\partial_{\phi''})E(\mc{L})=-3 E(\mc{L}),
\end{align}
 which vanishes when $E(\mc{\mc{L}})=0$, that is, when the field equation \eqref{E:phi4_1D} (with $m=0$) holds, and $X_2$ is hence a symmetry of the equation.
 In 4d, the corresponding symmetry will be strictly variational, cf.~\eqref{E:strictVarSymX11}.
	\subsubsection{Free, massive scalar theory}
\label{sec:FreeMassiveScalarTheory}
Now, if we let 
\begin{align}
	\lambda=0\quad \text{and} \quad m\ne 0
\end{align}
in \eqref{E:phi4_1D} and \eqref{E:detEqsPhi4end}, 
	the solutions to \eqref{E:detEqsPhi4end} are
\begin{align}
	\xi(x,\phi) &=d_1(x)+ \phi d_2(x), \nonumber \\
	\eta(x,\phi) &=d_3(x)+ \phi d_4(x)+\phi^2 d_2'(x), 
\end{align}
where 
\begin{align}
	d_1(x)&=c_1+\frac{c_2 \sin
   (2 m x)}{2 m}-\frac{c_4 \cos (2 m x)}{2 m^2}+\frac{c_4}{2 m^2}, \nonumber \\
	d_2(x)&= c_5 \cos (m x)+ c_6 \sin (m x),				\nonumber  \\
	d_3(x)&= c_7 \cos (m x)+ c_8 \sin (m x),  \nonumber \\
	d_4(x)&= -c_2 \sin ^2(m x)+c_3 +\frac{c_4 \sin (2 m x)}{2 m}.
\end{align}
Then the infinitesimal generator \eqref{E:infGenPoint} is given by
\begin{align}
	X= \sum_{i=1}^8 c_i X_i,
\end{align}
where the generators $X_i$ may be taken as
\begin{align}\label{E:generators2MST}
	X_1 &= (1/m) \partial_x, \nonumber \\
X_2 &= {(1/m) \sin (m x) \cos (m x)  {\partial_x}}-\sin ^2(m x) \phi   { \partial_\phi }, \nonumber \\
X_3 &= \phi   { \partial_\phi }, \nonumber \\
X_4 &= (1/m) \sin ^2(m x)  { \partial_x}+\sin (m x) \cos (m
   x) \phi   { \partial_\phi }, \nonumber \\
X_5 &=(1/m)\cos (m x) \phi   { \partial_x}- \sin (m x)
   \phi ^2  { \partial_\phi }, \nonumber \\
	X_6 &=  (1/m)\sin (m x) \phi   { \partial_x}+ \cos (m x)
   \phi ^2  { \partial_\phi }, \nonumber \\
	X_7 &= \cos (m x)  { \partial_\phi }, \nonumber \\
	X_8 &=\sin (m x) \partial_\phi,
\end{align}
where the factors $(1/m)$ ensure that there are no occurrences of the mass parameter $m$ in the commutator table.
Here, $X_1$ corresponds to translations and $X_3$ corresponds to scaling symmetries.
Some occurrences of the parameter $m^{-1}$ have been absorbed in the constants $c_j$, which should be kept in mind if we want to explore the symmetries corresponding to the case $m=0$. The commutator table of
 \eqref{E:generators2MST} can then be calculated by applying \texttt{SYM} \cite{dimas2005sym}
\begin{align}\label{T:comm_mne0}
\small{	\begin{array}{|c|c|c|c|c|c|c|c|c|}
\hline	
 \text{} & X_1 & X_2 & X_3 & X_4 & X_5 & X_6 & X_7 & X_8 \\ \hline
 X_1 & 0 & X_1-2 X_4 & 0 & 2 X_2+X_3 & -X_6 & X_5 & -X_8 & X_7 \\ \hline
 X_2 & 2 X_4-X_1 & 0 & 0 & X_4 & -X_5 & 0 & 0 & X_8 \\ \hline
 X_3 & 0 & 0 & 0 & 0 & X_5 & X_6 & -X_7 & -X_8 \\ \hline
 X_4 & -2 X_2-X_3 & -X_4 & 0 & 0 & -X_6 & 0 & -X_8 & 0 \\ \hline
 X_5 & X_6 & X_5 & -X_5 & X_6 & 0 & 0 & X_4-X_1 & X_3-X_2 \\ \hline
 X_6 & -X_5 & 0 & -X_6 & 0 & 0 & 0 & -X_2-2 X_3 & -X_4 \\ \hline
 X_7 & X_8 & 0 & X_7 & X_8 & X_1-X_4 & X_2+2 X_3 & 0 & 0 \\ \hline
 X_8 & -X_7 & -X_8 & X_8 & 0 & X_2-X_3 & X_4 & 0 & 0 \\ \hline
\end{array}}
\end{align}
Note that $\{X_1, X_7, X_8\}$, $\{X_1, X_5, X_6\}$  and $\{X_1, X_3, X_7, X_8\}$
are subalgebras. The subalgebra $\{X_3, X_6\}=\La{a}(1)$ is non-compact; thus,
the ambient algebra is non-compact as well.
We show that the Lie algebra generated by \eqref{E:generators2MST} is $\La{sl}(3)$ in Section \eqref{sec:MasslessScalarTheory}. The only scalar symmetry
$X_3$ has a prolongation that acts as
\begin{align}\label{E:pr1X3L1d}
	\on{pr}^{(1)}X_3(\mc{L})=(\phi \partial_\phi + \phi' \partial_{\phi'})(\mc{L})=2\mc{L},
\end{align}
and is hence not a variational symmetry, because the Euler operator of the result of \eqref{E:pr1X3L1d} is evidently not zero: $E(2\mc{L})=2E(\mc{L})$ is a multiple of the Euler-Lagrange-expression, which is not identically zero. Of course, the expression vanishes when assuming that the Euler-Lagrange equation holds; however we do not assume this on the level of variational (also called "off-shell") symmetries.
\subsubsection{Free, massless scalar theory}
\label{sec:MasslessScalarTheory}
In case  
\begin{align}
	m=\lambda=0
\end{align}
the solutions of the determining equations \eqref{E:detEqsPhi4end} still yield a Lie algebra that is 8d, with generators
\begin{align}\label{E:generatorsMasslessScalarTh}
	X_1 &=  \partial_x, \nonumber \\
X_2 &= x\partial_x, \nonumber \\
X_3 &= x^2\partial_x +x\phi \partial_\phi, \nonumber \\
X_4 &= \phi \partial_x, \nonumber \\
X_5 &=x \phi \partial_x +\phi^2 \partial_\phi, \nonumber \\
	X_6 &= \partial_\phi , \nonumber \\
	X_7 &= \phi \partial_\phi, \nonumber \\
	X_8 &= x\partial_\phi,
\end{align}
and commutator table
\begin{align}\label{T:commm=0}
\small{	\begin{array}{|c|c|c|c|c|c|c|c|c|}
	 \hline
 \text{} & X_1 & X_2 & X_3 & X_4 & X_5 & X_6 & X_7 & X_8 \\
 \hline X_1 & 0 & X_1 & 2 X_2+X_7 & 0 & X_4 & 0 & 0 & X_6 \\
 \hline X_2 & -X_1 & 0 & X_3 & -X_4 & 0 & 0 & 0 & X_8 \\
 \hline X_3 & -2 X_2-X_7 & -X_3 & 0 & -X_5 & 0 & -X_8 & 0 & 0 \\
 \hline X_4 & 0 & X_4 & X_5 & 0 & 0 & -X_1 & -X_4 & X_7-X_2 \\
 \hline X_5 & -X_4 & 0 & 0 & 0 & 0 & -X_2-2 X_7 & -X_5 & -X_3 \\
 \hline X_6 & 0 & 0 & X_8 & X_1 & X_2+2 X_7 & 0 & X_6 & 0 \\
 \hline X_7 & 0 & 0 & 0 & X_4 & X_5 & -X_6 & 0 & -X_8 \\
 \hline X_8 & -X_6 & -X_8 & 0 & X_2-X_7 & X_3 & 0 & X_8 & 0 \\
 \hline
\end{array} }
\end{align}
We again note that the algebra has an $\La{a}(1)=\{X_1,X_2\}$ subalgebra, so
the 8d algebra cannot be a compact algebra, because the subalgebra $\La{a}(1)$ is not compact. This algebra is $\La{sl}(3)$ (cf.~Olver's no 6.8 basis, e.g.~in \cite{laine2003classification}). 

Lie himself showed that a second-order equation with an 8d algebra, is always equivalent to the massless equation $\phi''=0$.
Equation \eqref{E:phi4_1D} with $\lambda=0, m^2 \ne 0$ is equivalent to the massless equation
$\hat{\phi}''=0$ through the point transformation
(see e.g.~\cite{mahomed2007symmetry})
\begin{align}\label{E:mapTomassless}
	\hat{x}= \tan(mx), \quad\hat{\phi}= \frac{\phi}{\cos(mx)}, 
\end{align}
which makes the massless equation hold if and only if the massive equation
holds, as long as $\cos(mx)\ne 0$. Note that the transformation \eqref{E:mapTomassless} is not a symmetry, because the structure of the equation is not conserved.

Moreover, the algebra corresponding to \eqref{E:generators2MST} and \eqref{T:comm_mne0} is $\La{sl}(3)$, which is the same algebra as in the case $m=0$. First, it is the only possible 8d algebra for any second-order ODE. Second,
the generators \eqref{E:generators2MST} may be transformed to generators 
yielding the same commutator table as for $m=0$ \eqref{T:commm=0}, by a basis shift:

By applying the transformations \eqref{E:mapTomassless} and their
inverses
\begin{align}\label{E:mapTomasslessInv}
x=\frac{\arctan(\hat{x})}{m},	\quad \phi=\frac{\hat{\phi}}{\sqrt{1+\hat{x}^2}}, 
\end{align}
in combination with chain rules
\begin{align}
	\partial_x &=(\partial_x \hat{x})\partial_{\hat{x}}+(\partial_x\hat{\phi})\partial_{\hat{\phi}}, \nn \\
	\partial_\phi &=(\partial_\phi \hat{x})\partial_{\hat{x}}+(\partial_\phi \hat{\phi})\partial_{\hat{\phi}},
\end{align}
we obtain a correspondence between the generators \eqref{E:generators2MST} related to $m\ne 0$, and the generators \eqref{E:generatorsMasslessScalarTh} corresponding to $m= 0$. Here, we will denote the latter as hatted quantities, consistent with the transformation \eqref{E:mapTomassless}. For instance,
\begin{align}
	X_1 &=(1/m)\partial_x=\frac{1}{\cos^2(mx)}(\partial_{\hat{x}}+\phi\sin(mx)\partial_{\hat{\phi}}) \nn \\
	&=(1+\hat{x}^2)(\partial_{\hat{x}}+\frac{\hat{x}\hat{\phi}}{1+\hat{x}^2}\partial_{\hat{\phi}}) \nn \\
	&=\partial_{\hat{x}}+ \hat{x}^2\partial_{\hat{x}}+\hat{x}\hat{\phi}\partial_{\hat{\phi}}\nn \\
	&=\hat{X}_1 +\hat{X}_3,
\end{align}
where the last line refers to the generators of \eqref{E:generatorsMasslessScalarTh}.
Subsequently, the generators of \eqref{E:generators2MST} and \eqref{E:generatorsMasslessScalarTh}
are connected through 
\begin{align}
	X=T \hat{X},
\end{align}
where the matrix $T$ equals
\begin{align}
	T=\left(
\begin{array}{cccccccc}
 1 & 0 & 1 & 0 & 0 & 0 & 0 & 0 \\
 0 & 1 & 0 & 0 & 0 & 0 & 0 & 0 \\
 0 & 0 & 0 & 0 & 0 & 0 & 1 & 0 \\
 0 & 0 & 1 & 0 & 0 & 0 & 0 & 0 \\
 0 & 0 & 0 & 1 & 0 & 0 & 0 & 0 \\
 0 & 0 & 0 & 0 & 1 & 0 & 0 & 0 \\
 0 & 0 & 0 & 0 & 0 & 1 & 0 & 0 \\
 0 & 0 & 0 & 0 & 0 & 0 & 0 & 1 \\
\end{array}
\right),
\end{align}
and we can express the generators \eqref{E:generators2MST}
in a new basis given by
\begin{align}
	{X}'= T^{-1}X.
\end{align}
Then, a calculation of the commutation table in the basis ${X}'$, shows that the table is identical to that in \eqref{T:commm=0}, the $m=0$ case. Hence, we have shown explicitly how the Lie symmetry algebra of the free, massive equation
is the same as that of the corresponding massless equation, namely $\La{sl}(3,\N{R})$.

As in the massive case, the generator $X_7=\phi \partial_\phi$ will correspond to a non-variational symmetry because its prolongation will again map $\mc{L}$ to a multiple of itself. The new scalar symmetry $X_6=\partial_\phi$ equals its own prolongation and annihilates
$\mc{L}=(1/2)(\phi')^2$, and is thus a strict variational symmetry. The accordance with Theorem \ref{P:prXannTannV}  is trivial, because $V=0$.
\subsection{$\phi^4$ theory in (3+1)d spacetime}
\label{sec:Phi4TheoryIn31DSpacetime}
Real scalar $\phi^4$ theory in 4d has a Lagrangian density
\begin{align}
	\mathcal{L} = \frac{1}{2} \partial_\mu \phi \, \partial^\mu \phi - \frac{1}{2} m^2 \phi^2 - \frac{\lambda}{4!} \phi^4,
\end{align}
with a nonlinear field equation
	\begin{align}\label{E:phi4_4D}
  \partial_\mu \partial^\mu \phi + m^2 \phi +\frac{\lambda}{6} \phi^3=0,
\end{align}
for $\phi=\phi(x)$. We may now analyze the symmetries of \eqref{E:phi4_4D}, for example by applying \texttt{MathLie} \cite{baumann2013symmetry} or \texttt{SYM} \cite{dimas2005sym}. 
Equation \eqref{E:phi4_4D} then yields 19 determining equations.
\subsubsection{The case $m,\lambda\ne 0$}
\label{sec:TheCaseMLambdaNe0}
In the case $m,\lambda \ne 0$, the Lie symmetry algebra consists of 10 generators, the first six of which correspond to the Lorentz algebra $\La{so}(1,3)$,
\begin{align}\label{E:LorentzAlg}
	X_1 &= x_0 \partial_{x_1}  + x_1 \partial_{x_0},  \nn \\
	X_2 &= x_0 \partial_{x_2}  + x_2 \partial_{x_0},  \nn \\
	X_3 &= x_0 \partial_{x_3}  + x_3 \partial_{x_0},   \nn \\
	X_4 &=  x_2 \partial_{x_3} - x_3 \partial_{x_2},     \nn \\
	X_5 &=  x_3 \partial_{x_1} - x_1 \partial_{x_3},    \nn \\
X_6 &=   x_1 \partial_{x_2}  -x_2 \partial_{x_1},
\end{align}
where the first three generators are the boosts in $x_1, x_2$ and $x_3$ directions, and the last three are the rotations about the $x_1, x_2$ and $x_3$ axes, respectively.
Together with the four translations
\begin{align}\label{E:PoincareTrans}
  		X_7 &= \partial_{x_0},\nn \\
	X_8 &=  \partial_{x_1},  \nn \\
	X_9 &= \partial_{x_2},  \nn \\
	X_{10} &=\partial_{x_3},
\end{align}
\eqref{E:LorentzAlg} yields the Poincaré algebra $\La{iso}(1,3)$, which is the symmetry algebra for the case $m,\lambda\ne 0$.
\subsubsection{The case $m=0$ and $\lambda\ne 0$}
\label{sec:TheCaseMu0lne0}
In the case of $\phi^4$ theory in (3+1)d spacetime with $m=0$ and $\lambda\ne 0$
\texttt{SYM} \cite{dimas2005sym} reveals that the symmetry algebra $\La{iso}(1,3)$ is enhanced with the five generators\footnote{Thanks to Stylianos Dimas for providing a tutorial Mathematica notebook demonstrating equations \eqref{E:LorentzAlg}--\eqref{E:phi^4_(3+1)d_mu=0} and
\eqref{E:phi^4_(3+1)d_lambda=0}--\eqref{E:arbFu4d} using \texttt{SYM}.}
\begin{align}\label{E:phi^4_(3+1)d_mu=0}
	X_{11} &= \Sigma_{i=1}^4 x_i\partial_{x_i} -\phi \partial_\phi         \nn \\
	X_{12} &= - 2 \phi x_{0} \, \partial_{\phi} + 2 x_{0} x_{1} \, \partial_{x_1} + 2 x_{0} x_{2} \, \partial_{x_2} + 2 x_{0} x_{3} \, \partial_{x_3} + (x_{0}^2 + x_{1}^2 + x_{2}^2 + x_{3}^2) \, \partial_{x_0},        \nn \\
	X_{13} &= - 2 \phi x_{1} \, \partial_{\phi} + 2 x_{0} x_{1} \, \partial_{x_0} + 2 x_{1} x_{2} \, \partial_{x_2} + 2 x_{1} x_{3} \, \partial_{x_3} + (x_{0}^2 + x_{1}^2 - x_{2}^2 - x_{3}^2) \, \partial_{x_1},        \nn \\
	X_{14} &=   - 2 \phi x_{2} \, \partial_{\phi} + 2 x_{0} x_{2} \, \partial_{x_0} + 2 x_{1} x_{2} \, \partial_{x_1} + 2 x_{2} x_{3} \, \partial_{x_3} + (x_{0}^2 - x_{1}^2 + x_{2}^2 - x_{3}^2) \, \partial_{x_2},     \nn \\
	X_{15} &=  - 2 \phi x_{3} \, \partial_{\phi} + 2 x_{0} x_{3} \, \partial_{x_0} + 2 x_{1} x_{3} \, \partial_{x_1} + 2 x_{2} x_{3} \, \partial_{x_2} + (x_{0}^2 - x_{1}^2 - x_{2}^2 + x_{3}^2) \, \partial_{x_3}.              
\end{align}
 
Scaling symmetries such as $X_{11}$ are known not to be preserved in the quantum theory, due to the regulator introducing a scale. 
We now demonstrate that the scaling symmetry generated by $X_{11}$ is, nevertheless, strictly variational, in contrast to the situation in the 1d case. The generator $X_{11}$ corresponds to infinitesimal symmetry transformations
\begin{align}\label{E:xAndphiScaleTrafoPhi4}
	x^\mu\to (1+\epsilon)x^\mu, \quad \phi\to (1-\epsilon)\phi.  
\end{align}
Then the derivative transforms as 
\begin{align}
	\frac{\partial \phi}{\partial x^\mu}\to \frac{(1-\epsilon)\partial \phi}{(1+\epsilon)\partial x^\mu}=(1-2\epsilon)\frac{\partial \phi}{\partial x^\mu},
\end{align}
while the measure transforms as
\begin{align}
	d^4x\to (1+4\epsilon)d^4x
\end{align}
to first order in the infinitesimal parameter $\epsilon$, cf.~\eqref{E:xAndphiScaleTrafoPhi4}.
Then, suppressing $\epsilon$, $\delta \phi =-\phi$, $\delta x= x$,
$\delta ({\partial_\mu \phi})= -2{\partial_\mu \phi}$ and $\delta (d^4x)=4d^4x$, which gives a first-order variation 
\begin{align}
	\frac{\delta (\mc{L}\, d^4x)}{d^4x}= \frac{\partial \mc{L}}{\partial \phi}\delta \phi+
	\frac{\partial \mc{L}}{\partial (\partial_\mu \phi)}\delta ({\partial_\mu \phi})+\frac{\mc{L}\delta (d^4x)}{d^4x}
	=\frac{\lambda \phi^4}{6}-2\partial_\mu \phi \partial^\mu \phi+4\mc{L}=0,
\end{align}
of the Lagrangian of $\phi^4$ theory.
We obtain the same result by applying the 1-prolongation of $X_{11}$, 
\begin{align}
	\operatorname{pr}^{(1)}X_{11}=x^\mu \partial_\mu-\phi \partial_\phi
	-2(\partial_\mu \phi) \frac{\partial}{\partial(\partial_\mu \phi)},
\end{align}
to $\mc{L}$, because the definition \eqref{E:condStrictVarSym} of a strict variational symmetry then holds, 
\begin{align}\label{E:strictVarSymX11}
	\operatorname{pr}^{(1)}X_{11}(\mc{L})+\mc{L}\, d_\mu\xi^\mu=-4\mc{L}+4\mc{L}=0,
\end{align}
where $\xi^\mu=x^\mu$, in contrast to the 1d case \eqref{E:prX21d}. There are no scalar symmetries, as in the case of the corresponding 1d theory in Section \ref{sec:MasslessPhi4Theory}. This is in accordance with Corollary \ref{C:1d4dcorr}, which states that scalar variational symmetries will be the same for the 1d and 4d theories.
\subsubsection{The case $m\ne 0$ and $\lambda=0$}
\label{sec:TheCaseMune0l0}
In the case of $\phi^4$ theory in (3+1)d spacetime with $m \ne 0$ and $\lambda = 0$
\texttt{SYM} \cite{dimas2005sym} shows the symmetry algebra $\La{iso}(1,3)$ is enhanced by the generators
\begin{align}\label{E:phi^4_(3+1)d_lambda=0}
	X_{11} &= \phi \partial_\phi, \nn \\
	X_{12} &= \mathcal{F}_{1}(x) \, \partial_{\phi},
\end{align}
where $\mathcal{F}_{1}(x)$ is an arbitrary function of the spacetime variables, satisfying 
\begin{align}\label{E:arbFu4d}
	\partial_\mu\partial^\mu \mathcal{F}_{1}+m^2\mathcal{F}_{1}=0.
\end{align}
Hence the symmetry algebra in this case becomes infinite dimensional, with $\La{iso}(1,3)$ as a finite subalgebra. The only scalar symmetry $X_{11}$ here is the same as for the 1d case in Section \ref{sec:FreeMassiveScalarTheory}. By
Corollary \ref{C:1d4dcorr}, the scalar, variational symmetries of the 4d theory must be the same as those of the 1d theory, which can easily be confirmed: As for the 1d case,
the prolongation of $X_{11}$ maps $\mc{L}$ to a multiple of itself, and is hence again non-variational. 

\subsubsection{The case $m=\lambda=0$}
\label{sec:TheCaseMu0l0}
In the case $m=\lambda=0$, the Poincaré algebra $\La{iso}(1,3)$ is enhanced by both sets of generators given in 
\eqref{E:phi^4_(3+1)d_mu=0} and \eqref{E:phi^4_(3+1)d_lambda=0}; hence 
the algebra once again becomes infinite-dimensional. The symmetry $X_{11}= \phi \partial_\phi$ will still map
$\mc{L}$ to a multiple of itself, as in Section \ref{sec:MasslessScalarTheory}, and is hence non-variational as before.
 The 4d, scalar symmetry is thus of the same nature as in the 1d case, in harmony with Corollary \ref{C:1d4dcorr}.

No divergence symmetries were found in the analysis in Section \ref{sec:ASimpleExamplePhi4Theory}, but they
occur e.g.~in scalar gauge singlet extensions of the SM \cite{singletExtensionsSM}.
Finally, we also note that all three cases, $m=0$, $\lambda=0$ and $\lambda=m=0$ may be imposed by symmetries, as all cases correspond to distinct symmetries.

\section{2HDM}
\label{sec:2HDM}
The Lagrangian of a 2HDM
is defined by \eqref{E:L_NHDM+KS} with $N=2$ and $K=0$,
\begin{align}\label{E:2HDMlag}
	\mc{L}
= -\tfrac{1}{4}W^a_{\mu\nu}W^{a\mu\nu}
  -\tfrac{1}{4}B_{\mu\nu}B^{\mu\nu}
  + \sum_{i=1}^2 (D_\mu\Phi_i)^\dagger (D^\mu\Phi_i)
  - V(\Phi_1,\Phi_2),
\end{align}
with 
Higgs doublets
\begin{align}\label{E:HiggsDoublets2HDM}
\Phi_1 = \frac{1}{\sqrt{2}}
\begin{pmatrix}
\phi_{1} + i \phi_{2} \\
\phi_{3} + i \phi_{4}
\end{pmatrix}, \quad
 \Phi_2 = \frac{1}{\sqrt{2}}
\begin{pmatrix}
\phi_{5} + i \phi_{6} \\
\phi_{7} + i \phi_{8}
\end{pmatrix},
\end{align}
where the
covariant derivatives and gauge boson field strength tensors are given by \eqref{E:covarDerDef}-\eqref{E:Bfst}.
The most general, renormalizable 2HDM potential can be written as
\begin{align}
V(\Phi_1,\Phi_2) &= m_{11}^2\,\Phi_1^\dagger\Phi_1 + m_{22}^2\,\Phi_2^\dagger\Phi_2
   -\big[ m_{12}^2\,\Phi_1^\dagger\Phi_2 + \text{h.c.}\big] \nonumber\\
 &\quad + \tfrac{\lambda_1}{2}(\Phi_1^\dagger\Phi_1)^2
   + \tfrac{\lambda_2}{2}(\Phi_2^\dagger\Phi_2)^2
   + \lambda_3(\Phi_1^\dagger\Phi_1)(\Phi_2^\dagger\Phi_2)
   + \lambda_4(\Phi_1^\dagger\Phi_2)(\Phi_2^\dagger\Phi_1) \nonumber\\
 &\quad + \Big[ \tfrac{\lambda_5}{2}(\Phi_1^\dagger\Phi_2)^2
   + \lambda_6(\Phi_1^\dagger\Phi_1)(\Phi_1^\dagger\Phi_2)
   + \lambda_7(\Phi_2^\dagger\Phi_2)(\Phi_1^\dagger\Phi_2)
   + \text{h.c.}\Big] ,
\end{align}
 where "h.c." means Hermitian conjugate, 
 i.e.~conjugate transpose, with complex parameters $m_{12}^2$, $\lambda_5$, $\lambda_6$ and $\lambda_7$, other parameters real. Applying a basis shift (i.e.~a reparametrization) of the doublets,
\begin{align}\label{E:basisTrafo2HDM}
	(\Phi_1,\Phi_2)^T\to(\hat{\Phi}_1,\hat{\Phi}_2)^T = U (\Phi_1,\Phi_2)^T
\end{align}
 with\footnote{Or with $U\in\Lg{SU}(2)$, if we absorb the complex phase of $\Lg{U}(2)$ into $\Lg{U}(1)_Y$ hypercharge symmetry.} $U\in \Lg{U}(2)$, we can diagonalize the mass-squared matrix
and remove the complex phase of e.g.~$\lambda_5$, which means that 
\begin{align}\label{E:parsDiagMM}
	m_{12}^2=0, \qquad \Im(\lambda_5)=0,
\end{align}
 in the potential $\hat{V}(\hat{\Phi}_1,\hat{\Phi}_2)=V({\Phi}_1,{\Phi}_2)$ written in the new, hatted basis.

The 24 Euler-Lagrange equations $E(\mc{L})=0$ in \eqref{E:2HDMlag}
are of the form
\begin{align}\label{E:EL2HDM}
	\frac{\partial \mc{L}}{\partial y^i}-d_\mu \frac{\partial \mc{L}}{\partial y^i_{,\mu}}=0, \quad 1\leq i \leq 24,
\end{align}
with 
\begin{align}\label{E:24fields}
	\{y^i\}_{i=1}^{24}=\{\phi_1,\ldots, \phi_8,B_0,\ldots,B_3, W_0^1,W_0^2,\ldots,W_3^3\}.
\end{align}
These equations will be variants of the interacting (i.e.,~nonlinear) Klein-Gordon and Proca equations.

\subsection{A bilinear formalism}
\label{sec:ABilinearFormalism}
The 2HDM potential may also be written in a particularly useful manner by gauge invariant bilinears, a formalism first introduced in \cite{Nagel:2004sw} and subsequently
developed for the general 2HDM in \cite{Maniatis:2007vn}.
\begin{equation}
\label{eq:Vbilinear}
V = M_0r_0 + M_a r_a + \Lambda_0 r_0^2 + L_a r_0r_a + r_a \Lambda_{ab} r_b
\end{equation}
by employing bilinears in the fields,
\begin{equation}\label{E:bilinears}
r_0 = \Phi_i^\dag \Phi_i \;, \quad r_a = \Phi_i^\dagger (\sigma_a)_{ij} \Phi_j,
\end{equation}
where $a \in \{1,2,3\}$ and bilinears $r_a$ are given in terms of the Pauli matrices $\sigma_a$, which imply
\begin{align}
	r_1 &= \frac{1}{2}(\Phi_1^\dag \Phi_2+ \Phi_2^\dag \Phi_1), \quad
		r_2 = -\frac{i}{2}(\Phi_1^\dag \Phi_2 - \Phi_2^\dag \Phi_1), \nn \\
			r_3 &= \frac{1}{2}(\Phi_1^\dag \Phi_1- \Phi_2^\dag \Phi_2).
\end{align}
The parameters of \eqref{eq:Vbilinear} are given by
\begin{align}
	M_0&=m_{11}^2+m_{22}^2, \quad \Lambda_0= \tfrac{1}{2}(\lambda_1+\lambda_2)+\lambda_3, \\
	L&=\left(-\Re(\lambda_6+\lambda_7), \Im(\lambda_6+\lambda_7),\frac{1}{2}(\lambda_2 -\lambda_1) \right)^T \\
 M&=\left(2\Re(m_{12}^2),-2\Im(m_{12}^2), m_{22}^2-m_{11}^2\right)^T, \\
	\Lambda &= 
	\left( \begin{array}[pos]{ccc}
	\lambda_4 +\Re(\lambda_5)	& -\Im(\lambda_5) & \Re(\lambda_6-\lambda_7) \\
	-\Im(\lambda_5)	& \lambda_4 -\Re(\lambda_5)& -\Im(\lambda_6-\lambda_7)\\
\Re(\lambda_6-\lambda_7)		& -\Im(\lambda_6-\lambda_7) & \frac{1}{2}(\lambda_1+\lambda_2)-\lambda_3
	\end{array} \right). \label{E:Lambda2HDM}
\end{align}
Under a change of basis
\begin{equation}\label{E:HbasisCh}
\Phi_i \rightarrow U_{ij}\Phi_j\, , \quad  U\in \Lg{SU}(2),
\end{equation}
 $r_0$ is a singlet while $r_a$ transforms under the adjoint representation of $\Lg{SU}(2)$,
\begin{equation}
r_0 \rightarrow r_0 \,, \quad r_a\rightarrow R_{ab}(U)r_b,
\end{equation}
with
\begin{equation}\label{E:R(U)}
R_{ab}(U) = \frac{1}{2}\text{Tr}(U^\dagger \lambda_a U \lambda_b).
\end{equation}
To keep the potential invariant, 
the coupling constants are transformed as follows under this change of basis:
\begin{align}
\label{E:LambdaTrafoU}
\Lambda &\to R(U) \Lambda R^T(U),\\
L &\to R(U) L, \\
M &\to R(U) M. 
\end{align}

Because $\text{Ad}_{\Lg{SU}(2)}=\Lg{SO}(3)$ all matrices of $\Lg{SO}(3)$ correspond to an $\Lg{SU}(2)$ Higgs basis transformation, and we may always diagonalize the matrix $\Lambda$ of \eqref{E:Lambda2HDM} by a basis transformation. Given an arbitrary potential, we can change the basis in such a way that
$\Lambda$
is diagonalized. This means that 
\begin{align}\label{E:diagLambdapars}
	\Im(\lambda_5)=0, \quad \lambda_7=\lambda_6,
\end{align}
in the new basis. The basis shift will, of course, depend on the values of the parameters of $\Lambda$ in the original basis, hence we have reduced the number of parameters by three.
A slightly different bilinear formalism was applied in the beautiful derivation of the strict variational symmetries of the 2HDM in
\cite{PhysRevD.77.015017}. A crucial part of the proof is the aforementioned diagonalizability of $\Lambda$, which does not carry over to the analogous matrix $\Lambda$
 for higher $N$, because $\text{Ad}_{\Lg{SU}(N)}\varsubsetneq \Lg{SO}(N^2-1)$ is dramatically smaller when $N>2$. This again makes the classification of symmetries using bilinear formalisms much harder for $N>2$.

\subsection{Finding and solving determining equations}
\label{sec:FindingAndSolvingTheDeterminingEquations}
We now find the determining equations 
\eqref{E:prDelta=0PDEsymCond} of the system of PDEs given by \eqref{E:EL2HDM}
by applying \texttt{SYM} \cite{dimas2005sym}. This yields a system of 1733
determining equations for the coefficients $\xi^\mu$ and $\eta^i$ of the
point-symmetry generator \eqref{E:infGenPoint}.

These determining equations follow from the symmetry condition
\eqref{E:prDelta=0PDEsymCond}, that is, by computing
$\pr X\big(E_i(\mathcal{L})\big)$ for all $i$ as functions on
the relevant jet space and then imposing $E(\mathcal{L})=0$ by substituting
24 of the highest-order derivatives, cf.~\eqref{E:substPhi4}. In geometric
terms, this ensures that the change of the functions $E_i(\mathcal{L})$
along the direction $\pr X$ vanishes when starting on the solution
manifold $\mc{M}_{E(\mc{L})}$, so that $\pr X$ is
tangent to $\mc{M}_{E(\mc{L})}$. Hence $X$ generates a Lie point
symmetry, since it moves us infinitesimally from one solution to another.
Finally, by demanding that the coefficients of all linearly independent
monomials in the remaining derivatives of the fields vanish, we obtain a
linear overdetermined system of PDEs for the functions $\xi^\mu$ and
$\eta^i$, cf.~\eqref{E:detEqsPhi4end}.

As we are interested only in scalar symmetries, we set
\begin{align}
	\xi^\mu&=0,\quad \text{for all}\quad 0\leq \mu \leq 3, \nn \\
	\eta^i&=0,\quad \text{for all} \quad 9\leq i \leq 24.
\end{align}
This means that the spacetime variables and gauge bosons are kept constant under the considered transformations. Then the simplest (shortest) of the remaining determining equations are
\begin{align}\label{E:onlyAffineSym}
	\partial_{\phi_j} \partial_{\phi_k} \eta^i =0,\quad \text{for all}\quad 1\leq i,j,k \leq 8,
\end{align}
and
\begin{align}
	\partial_{y^j} \eta^i =0 \quad \text{for all}\quad j>8 \wedge 1\leq i \leq 8, 
\end{align}
which means that the $\eta^i$'s are affine (or sometimes referred to as linear) in the scalar fields and do not depend on the gauge fields. 
For strict variational symmetries, \eqref{E:onlyAffineSym} is natural:
any quadratic or higher dependence
$\eta^i = \mathcal{O}(\phi^2)$ would generate higher-order monomials
(such as $\phi^5$ or $\phi^8$) in the transformed potential, which would spoil the invariance of the potential.
   
Hence, we can proceed by solving the determining equations by substituting
\begin{align}\label{E:2HDMansatz}
	\eta^i= a_i+b_{ij}\phi_j,
\end{align}
  with an implicit sum over $j$ ranging from 1 to 8 into the determining equations. We do not include explicit spacetime variables
	in \eqref{E:2HDMansatz} because we are only considering pure scalar symmetries, although it would have been a principal possibility for non-scalar symmetries. After substituting \eqref{E:2HDMansatz}, we obtain a system of 353 polynomial
equations
\begin{align}\label{E:pol2HDM}
  P_k(y^1,\ldots,y^{q}) = 0,
\end{align}
for the $q=24$ fields \eqref{E:24fields}, with $1 \leq k \leq 353$.
All derivatives of the fields were eliminated when deriving the general
determining equations for the functions $\eta^i$ and $\xi^\mu$. The
polynomials $P_k$ depend on the parameters of the Lagrangian, as well as on
the unknown constants $a_i$ and $b_{ij}$, which we wish to determine for the
different possible choices of Lagrangian parameters. For a symmetry to be
present, all equations \eqref{E:pol2HDM} must hold for all field values, in
complete analogy with the requirement that \eqref{E:detEqsPhi4end} must hold for
all values of $\phi$ (and $x$). This implies that, in each equation
\eqref{E:pol2HDM}, the coefficients of all distinct monomials in the fields
must vanish.

	More precisely, let a general monomial $\bar{m}$ in the $q=24$ fields  be written as:
	\begin{align}
		\bar{m}(n_1,\ldots,n_q)=(y^1)^{n_1}\cdots (y^{q})^{n_{q}},\quad n_j \in \N{N}_0,
	\end{align}
	where $\N{N}_0=\N{N}\cup\{0\}$.
	Then
	\begin{align}\label{E:coeffMonom=0}
		C \bar{m}(n_1,\ldots,n_q)\in P_k \Rightarrow C=0
		\end{align}
		when $C\equiv C_{n_1,\ldots,n_q}(a,b,m^2,\lambda,g,g'Y)$ is the maximal coefficient of $\bar{m}(n_1,\ldots,n_q)$, that is  
\begin{align}
	D \bar{m}(n_1,\ldots,n_q)\notin (P_k-C \bar{m}(n_1,\ldots,n_q)),
\end{align}
	for any, non-zero coefficient $D\equiv D_{n_1,\ldots,n_q}(a,b,m^2,\lambda,g,g'Y)$.	
This splitting into monomials yields 1412 equations of the form $C=0$, some of which are very simple.
	The 80 simplest equations are of this type: 
	\begin{align}\label{E:80simp2HDM}
		h\, a_i=0,\quad\text{and}\quad h\, b_{ij}=0, \quad \text{with}\quad h\in \{g,g'Y,g g'Y\},
	\end{align}
	 and may be solved with solutions 
	\begin{align}\label{E:a_i02HDM}
	a_i=0 	\quad \forall i \in \{1,\ldots, 8\},
	\end{align}
	and 
	\begin{align}\label{E:bijsubst12HDM}
	b_{ij}=0
	\end{align}
	for all $i,j$ except for indices corresponding to any of the following 24 parameters
	\begin{align}\label{E:B1_nozerobs2HDM}
	B_1=	\{&b_{1 2},\, b_{1 5},\, b_{1 6},\, b_{2 1},\, b_{2 5},\, b_{2 6},\, b_{3 4},\, 
 b_{3 7},\, b_{3 8},\, b_{4 3},\, b_{4 7},\, b_{4 8},\nn \\
&b_{5 1},\, b_{5 2},\, 
 b_{5 6},\, b_{6 1},\, b_{6 2},\, b_{6 5},\, b_{7 3},\, b_{7 4},\, b_{7 8},\, 
 b_{8 3},\, b_{8 4},\, b_{8 7}\},
	\end{align}
	parameters that did not vanish at this stage.
	We also could have considered the case $g'=0$ here, and hence investigated custodial symmetries \cite{Sikivie:1980hm}, but we will refrain
	from it, as these are not exact symmetries of the considered Lagrangian \eqref{E:2HDMlag}, and as this could have doubled the analysis.\footnote{See \cite{Grzadkowski:2010dj} for conditions for the canonical
custodial symmetry (CCS) in the 2HDM, while \cite{Pilaftsis:2011ed}
classifies non-canonical 2HDM symmetries. The Lie symmetry analysis used in
this work would detect all such symmetries, including the non-canonical
ones.}

 Equation \eqref{E:a_i02HDM} means that there cannot be any pure scalar shift symmetries of the field equations in a 2HDM. 
	We continue our process of solving the determining equations by substituting the solutions \eqref{E:a_i02HDM} and \eqref{E:bijsubst12HDM} into the 1412 determining equations, where only the parameters \eqref{E:B1_nozerobs2HDM} are kept non-zero among the $b_{ij}$'s. 
	
	At this stage, we also implement \eqref{E:parsDiagMM}, that is, set $m_{12}^2=0$ and $\Im(\lambda_5)=0$, although we just as well could have done so from the start, before calculating the Euler-Lagrange equations. Then, we end up with a system
	\begin{align}\label{E:systemRed1}
 \mc{D}_i=0		
	\end{align}
	of 548 determining equations.
	
\subsection{Parameter cases and reductions of the potential}
\label{sec:ParameterCasesAndReductionsOfThePotential}

Let a \textit{reduced potential} be a potential in which one or more parameters
have been eliminated by transforming to a specific scalar basis through a
$\Lg{U}(2)$ Higgs-basis transformation \eqref{E:basisTrafo2HDM}. Such reduced potentials have been used, for example, in early studies of CP
violation in the 2HDM \cite{Lee:1973iz}, in basis-independent analyses
\cite{Davidson:2005cw}, in the Minkowski-space formulation of the Higgs
potential \cite{PhysRevD.77.015017}, and in tensor-product-based
classifications of scalar symmetries \cite{Pilaftsis:2011ed}.

  We proceed by considering different parameter cases, and try to optimally reduce the potential in each case.
This means that we eliminate as many parameters as possible through an appropriate basis transformation, and hence, more or less, lock the basis and thereby discard reparametrization freedom. We do this to avoid that the same symmetry manifests itself in different bases of the doublets excessively, typically by Lie algebras with generators involving potential parameters. Although all these manifestations of the same symmetry will be detected by the Lie symmetry analysis (for a clarification of this, see Appendix \ref{sec:ChangesOfVariables}), we want to avoid this as much as possible, because we would otherwise have to show that two different manifestations of the same symmetry are equivalent, which may imply extra work. 
	
\subsubsection{The case $m_{11}^2\ne m_{22}^2$}
\label{sec:TheCaseM112NeM222}
In addition to the diagonalization of the mass-squared matrix, and hence the reduction \eqref{E:parsDiagMM} of the potential, we now assume that
\begin{align}
	m_{11}^2\ne m_{22}^2
\end{align}
Then the 164 simplest determining  
equations are of the forms 
\begin{align}
	(m_{11}^2-m_{22}^2) b_{ij} &= 0 \Rightarrow b_{ij}=0  \\
	h(b_{ij}\pm b_{kl}) &=0 \Rightarrow b_{ij}=\mp b_{kl} \label{E:hbb=0}
\end{align}
 for $h>0$ cf.~\eqref{E:80simp2HDM} and certain indices $i,j,k$ and $l$. Solving these equations yields:
\begin{align}\label{E:bsolsm11ne22}
	-b_{12}&=-b_{34}=b_{43} = b_{21},  \nn \\
	b_{56}&=-b_{65}=-b_{87}= b_{78}, \nn \\
	b_{ij}&=0, \quad \text{for other}\;\: i,j, 	
\end{align}
which means that we have only two free parameters: $b_{21}$ and $b_{78}$.
Inserting \eqref{E:bsolsm11ne22} into the full set of determining equations and then applying \texttt{Mathematica}'s built-in \texttt{Reduce} function, we obtain 
\begin{align}\label{E:b1m11ne22}
	b_{21}+b_{78}=0 
\end{align}
 or 
\begin{align}\label{E:b2m11ne22}
		\Re(\lambda_5)=\Re(\lambda_6)=\Re(\lambda_7)=\Im(\lambda_6)=\Im(\lambda_7)=0
	\end{align}
Here, substituting \eqref{E:b1m11ne22}
 into \eqref{E:2HDMansatz} and then into
the infinitesimal generator \eqref{E:infGenPoint}
results in a 1-dimensional algebra that is present for all\footnote{This corresponds to the fact that there are no restrictions on the parameters of the potential in \eqref{E:b1m11ne22}.} potentials, given by the generator
\begin{align}
	X_Y= -\phi_2 \partial_{\phi_1}+\phi_1 \partial_{\phi_2}-\phi_4\partial_{\phi_3}+\phi_3\partial_{\phi_4}-\phi_6\partial_{\phi_5}+\phi_5\partial_{\phi_6}-\phi_8\partial_{\phi_7}+\phi_7\partial_{\phi_8}
\end{align}
which again implies the symmetry algebra 
\begin{align}
\La{u}(1)_Y=\on{span}(X_Y), 	
\end{align}
that is, the Lie algebra of the hypercharge symmetry group $\Lg{U}(1)_Y$.
The only other possibility \eqref{E:b2m11ne22} corresponds to a 2-dimensional algebra, parametrized by the two now free parameters 
$b_{21}, b_{78}$, because there are no restrictions on $b_{21}, b_{78}$ in \eqref{E:b2m11ne22}. Equations \eqref{E:2HDMansatz} and \eqref{E:infGenPoint} then yield a symmetry algebra
\begin{align}
	\La{u}(1)_Y \oplus \La{u}(1)_\text{PQ}=\on{span}(X_Y,X_\text{PQ}),
\end{align}
where
\begin{align}\label{E:PQgen2HDM}
X_\text{PQ}= \phi_2 \partial_{\phi_1}-\phi_1 \partial_{\phi_2}+\phi_4\partial_{\phi_3}-\phi_3\partial_{\phi_4}-\phi_6\partial_{\phi_5}+\phi_5\partial_{\phi_6}-\phi_8\partial_{\phi_7}+\phi_7\partial_{\phi_8},
\end{align}
which spans the Lie algebra $\La{u}(1)_\text{PQ}$ of the Peccei-Quinn $\Lg{U}(1)$ symmetry. We demonstrate that this is a strict variational symmetry by applying the fact 
$\on{pr}X_\text{PQ}(T)=0$ (i.e.,~the Peccei-Quinn symmetry is an SVS of the kinetic terms), together with Theorem \ref{P:prXannTannV}. Alternatively, we can show
\eqref{E:condStrictVarSym} holds for potential parameters \eqref{E:parsDiagMM} and \eqref{E:b2m11ne22}, or we can simply check the consistency with the results in \cite{PhysRevD.77.015017} or 
Table 5 of \cite{Branco_2012}. 

\subsubsection{The case $m_{11}^2= m_{22}^2$}
\label{sec:TheCaseM112M222}
We now assume
\begin{align}\label{m11sq=m22sq}
	m_{11}^2= m_{22}^2,
\end{align}
and substituting \eqref{m11sq=m22sq} in the determining equations \eqref{E:systemRed1}. The 148 simplest of these equations are of the form
\eqref{E:hbb=0}, and we solve and substitute the resulting 
equations in the total system of determining equations.
 The only free parameters $b_{ij}$ in the determining equations can now be taken as:
 \begin{align}\label{E:B22HDM}
B_2=\{b_{21},b_{25}, b_{51}, b_{78}\}.	 
 \end{align}
Because of \eqref{m11sq=m22sq}, the mass-squared matrix will remain diagonal under any additional Higgs basis transformations, hence, we have the freedom to further reduce the potential, without introducing a new $m_{12}^2$ parameter. 
Therefore, we choose to diagonalize the matrix $\Lambda$ in \eqref{E:Lambda2HDM}, 
which means that
\begin{align}\label{E:parsDiagLambda}
	\lambda_6=\lambda_7, \quad \Im(\lambda_5)=0,
\end{align}
cf.~\eqref{E:diagLambdapars}, where the latter reduction $\Im(\lambda_5)=0$ is the same as before, cf.~\eqref{E:parsDiagMM}.

\paragraph{Assuming $\Re(\lambda_5)\ne 0$}
\label{sec:AssumingReLambda5Ne0}
Additionally, we will now assume
\begin{align}\label{E:reL5ne0}
	\Re(\lambda_5)\ne 0,
\end{align}
 since $\Re(\lambda_5)=0$ would also have let us reduce the potential by eliminating $\Im(\lambda_6)$.
By substituting \eqref{E:parsDiagLambda} into the determining equations  
and
solving for the parameters \eqref{E:B22HDM},
 with assumption \eqref{E:reL5ne0}, we obtain three different solutions for the parameters $B_2$, cf.~\eqref{E:B22HDM}, hence three possible Lie symmetry algebras: Let 
\begin{align}
	H_1 &=\phi_6 
   \partial_{\phi_1}-\phi _5 
   \partial_{\phi_2}+\phi _8 
   \partial_{\phi_3}-\phi _7 
   \partial_{\phi_4}+\phi _2 
   \partial_{\phi_5}-\phi _1 
   \partial_{\phi_6}+\phi _4 
   \partial_{\phi_7}-\phi _3 
   \partial_{\phi_8}, \nn \\  
	H_2&=\phi _5 
   \partial_{\phi_1}+\phi _6 
   \partial_{\phi_2}+\phi _7 
   \partial_{\phi_3}+\phi _8 
   \partial_{\phi_4}-\phi _1 
   \partial_{\phi_5}-\phi _2 
   \partial_{\phi_6}-\phi _3 
   \partial_{\phi_7}-\phi _4 
   \partial_{\phi_8}.
\end{align}
Then, the Lie algebra of the reparametrization group $\Lg{SU}(2)$
(i.e.~the basis transformations), denoted $\La{su}(2)_\text{HF}$ (where "HF" stands for "Higgs Family"), will be spanned by the latter generators and 
\begin{align}
	H_3\equiv X_\text{PQ},
\end{align}
 c.f.~\eqref{E:PQgen2HDM}, namely
 \begin{align}
	 \La{su}(2)_\text{HF}= \on{span}(H_1,H_2,H_3).
 \end{align}
	The Lie algebra in this basis has the commutation rules\footnote{Remember we are employing the mathematicians' definition of a Lie algebra, and are hence working in an anti-Hermitian basis.}
	\begin{align}
	[H_i,H_{j}]=-2\epsilon^{ijk}H_{k}.	
	\end{align}
	 Now, the first solution for the $b$'s is
\begin{align}\label{E:solIbsconds}
\text{I}:\quad	b_{25}&=0,\quad b_{78}=-b_{21}, \nn \\
	\text{when}\quad \lambda_2 &= \lambda_3+\lambda_4+\Re(\lambda_5)\wedge \lambda_1= \lambda_2
\end{align}
with $b_{51}$ free, and this holds for the displayed parameter conditions, in addition to the other conditions as in \eqref{E:parsDiagMM}.
The other two solutions are as follows:
\begin{align}
&\text{II}:\quad	b_{51}=0, \quad b_{78}=-b_{21}
&\text{III}:\quad	b_{25}=b_{51}=0, \quad b_{78}=-b_{21}
\end{align}
 where both solutions hold for several, different conditions on the potential parameters, similar to those given in \eqref{E:solIbsconds}.
	Then, the three symmetry algebras are as follows:
	\begin{align}\label{E:Lasm11sq=m22sq}
		\La{g}_\text{I} &= \on{span}(H_2,X_Y), \\
		\La{g}_\text{II} &= \on{span}(H_1,X_Y),  \\
		\La{u}(1)_Y &= \on{span}(X_Y).
	\end{align}
  We can conclude that all these algebras are strict variational symmetries, without considering all the distinct sets of conditions on the parameters of the potential: 
 The prolonged infinitesimal elements $\on{pr}X$ of all generators $X$ in \eqref{E:Lasm11sq=m22sq} will annihilate
the kinetic terms when applied to the most general $\mc{L}$. This may be verified explicitly, but it also follows from the well-known fact that $\Lg{SU}(2)_\text{HF}$ is a strict variational symmetry group of the kinetic terms of the 2HDM, see \cite{Olaussen:2010aq} for a proof that
the strict variational symmetries of the kinetic terms are $\La{su}(N)$ for a general NHDM. Then, Theorem \ref{P:prXannTannV}
guarantees that these symmetries are strictly variational, independent of the parameter constraints of the potentials.
Moreover, in Section \ref{sec:Conclusion2HDM}, we will show that the symmetries given by $\La{g}_\text{I}$ and $\La{g}_\text{II}$ are equivalent, which means that there is a new Higgs basis where $\La{g}_\text{I}$ will be the same as $\La{g}_\text{II}$ in the old basis.


\paragraph{Assuming $\Re(\lambda_5)= 0$}
\label{sec:AssumingReLambda50}
We now consider the special case where
\begin{align}\label{E:Rel5=02HDM}
	\Re(\lambda_5)= 0,
\end{align}
in addition to the conditions \eqref{m11sq=m22sq} and \eqref{E:parsDiagLambda}.
By a rephasing of the doublets, we may now eliminate $\Im(\lambda_6)$,
 \begin{align}\label{E:Iml6=02HDM}
	\Im(\lambda_6)= 0,
\end{align}
without introducing any new parameters.

\subparagraph{Assuming $\Re(\lambda_5)= 0$ and $\lambda_1+\lambda_2\ne 2(\lambda_3+\lambda_4)$}
\label{sec:AssumingReLambda50AndLambda1Lambda2Ne2Lambda3Lambda4}
If $\lambda_1+\lambda_2= 2(\lambda_3+\lambda_4)$ the matrix $\Lambda$ will be proportional to identity, and we can perform $\Lg{SO}(3)$ rotations without altering the elements of $\Lambda$. Then, we can freely rotate $L$ while $M$ remains zero, such that $L$ points in a desired direction, for example, in the $x$-direction, and then will $\lambda_1=\lambda_2$. Hence, we first consider the case in which:
\begin{align}
	\lambda_1+\lambda_2\ne 2(\lambda_3+\lambda_4),
\end{align}
in addition to all the previous assumptions, such as \eqref{E:Rel5=02HDM} and \eqref{E:Iml6=02HDM}.
By solving the determining equations for this case, 
we obtain three solutions for the surviving parameters $B_2$, cf.~\eqref{E:B22HDM}. In the same manner as in earlier sections, the resulting Lie symmetry algebras are 
\begin{align}\label{E:algebrasRel5l1l2ne2l3l4=0}
	\La{g}_1 &=\on{span}(H_3,X_Y),  \\
	\La{u}(1)_Y &=\on{span}(X_Y),
\end{align}
where two different solutions for the $b$'s correspond to the same algebra $\La{u}(1)_Y$, while $\La{g}_1$ is isomorphic to to 
$\La{u}(1)\oplus \La{u}(1)_Y$.
In the same manner as before, all symmetries given by \eqref{E:algebrasRel5l1l2ne2l3l4=0} are, by Theorem \ref{P:prXannTannV}, strictly variational because they annihilate the kinetic terms of the 2HDM.  

\subparagraph{Assuming $\Re(\lambda_5)= 0$ and $\lambda_1+\lambda_2= 2(\lambda_3+\lambda_4)$}
\label{sec:AssumingReLambda50AndLambda1Lambda2ne2Lambda3Lambda4}
We now consider the last case with, as explained at the beginning of the last paragraph,
\begin{align}
	\lambda_1+\lambda_2= 2(\lambda_3+\lambda_4) \wedge \lambda_1=\lambda_2,
\end{align}
in addition to \eqref{m11sq=m22sq}, \eqref{E:parsDiagLambda}, \eqref{E:Rel5=02HDM} and \eqref{E:Iml6=02HDM}.
Under these conditions we obtain two solutions for the $b$'s. In the first solution, $\Re(\lambda_6)$ is free, and the corresponding symmetry algebra is
\begin{align}
	\La{g}_1= \on{span}(H_1 ,X_Y),
\end{align}
 which again is isomorphic to $\La{u}(1)\oplus \La{u}(1)_Y$.
The second solution is that all the four parameters of $B_2$ are free, cf.~\eqref{E:B22HDM}, while $\Re(\lambda_6)=0$. In this case, the symmetry algebra
is 
\begin{align}
	\La{su}(2)_\text{HF}\oplus
	\La{u}(1)_Y= \on{span}(H_1, H_2, H_3, X_Y).
\end{align}
Both symmetry algebras correspond to strict variational symmetries, due to Theorem \ref{P:prXannTannV}.
\subsubsection{The inequivalent Lie point symmetries of the 2HDM}
\label{sec:Conclusion2HDM}
All the 2-dimensional Lie algebras found in the analysis are equivalent, because a rotation 
\begin{align}\label{E:su2RotGen}
	H_1\to H_2 \to H_3 \to H_1 
\end{align}
of $\La{su}(2)_\text{HF}$ generators will be an inner automorphism of $\La{su}(2)_\text{HF}$, which means it may be implemented by a matrix conjugation by a matrix $U\in \Lg{SU}(2)$, that is, $UH_1 U^\dag =H_2$, etc. Moreover, a conjugation by a $\Lg{SU}(2)$ matrix $U$ corresponds to a change of Higgs doublet basis, which means all the 2-dimensional symmetries we found really are the same symmetry transformation, expressed in different Higgs bases.
Here the other element $X_Y$ of the 2-dimensional algebras corresponds to a complex phase, and commutes with all other (series of) generators, including $U$, and is not affected by the automorphism.
 Finally, the rotation of generators \eqref {E:su2RotGen} is an automorphism because it conserves the commutator, and it is inner because all automorphisms of $\La{su}(2)$ are inner (its Dynkin diagram is just a dot with no non-trivial diagram automorphisms).  
Hence, we have shown that there are only three possible, inequivalent scalar Lie point symmetry algebras of the 2HDM, namely
\begin{align}\label{E:conclusion2HDMsymmetries}
	\La{su}(2)\oplus \La{u}(1)_Y, \quad \La{u}(1)\oplus \La{u}(1)_Y,\quad \La{u}(1)_Y
\end{align}
and they are all strictly variational according to Theorem \ref{P:prXannTannV}, consistent with the results in 
\cite{PhysRevD.77.015017} and \cite{Pilaftsis:2011ed}. Gauge symmetries are not included in \eqref{E:conclusion2HDMsymmetries} because we considered only (purely) scalar symmetries. In addition, we have demonstrated something new, namely, that there are no scalar divergence or scalar non-variational Lie point symmetries in the 2HDM, 
as Lie's method finds all symmetries of systems of PDEs,
cf.~\eqref{E:prDelta=0PDEsymCond}. The absence of divergence and non-variational Lie point symmetries in the 2HDM stands in contrast to the situation in singlet extensions of the SM, where such symmetries do appear \cite{singletExtensionsSM}. 

Furthermore, we have tested the soundness of our implementation of Lie symmetry analysis for models with many parameters, variables and reparametrization freedom by applying the method to a well-understood example and replicating well-known results. Of course, the derivation in \cite{PhysRevD.77.015017} of the strict variational point symmetries of the 2HDM is more effective, transparent and elegant, and includes discrete symmetries in the same proof. However, the method applied in this section has the advantage that it is universal and detects all divergence and non-variational symmetries as well as strict variational symmetries. Lie symmetry analysis could, at least in principle, be applied to any model, including any NHDM, in contrast to methods that only consider the possible symmetry transformations of gauge invariant scalar bilinears, cf.~\eqref{E:bilinears}. The method applied in this study also yields the exact parameter conditions for all symmetries, although this may not be an essential feature in a classification of symmetries.

\section{Summary and outlook}
\label{sec:SummaryAndOutlook}

We have demonstrated, by example, how Lie symmetry analysis of field equations
can be applied to particle physics models with many variables, parameters, and
reparametrization freedom. In our analysis of the two-Higgs-doublet model
(2HDM), we reproduced its well-known strict variational Lie point symmetries
and, for the first time, showed that the model admits neither divergence nor
non-variational Lie point symmetries. Hence, the maximal realizable Lie point symmetry algebras of the 2HDM are
generated exhaustively by the strict variational symmetries.
This observation is essential: Under the usual boundary conditions, divergence
symmetries are symmetries of the action and, provided that the path-integral
measure is invariant, they should be preserved under quantization in
the same way as strict variational symmetries.
  
Moreover, we proved three general results that can simplify 
  Lie symmetry analysis for a wide class of particle physics models. The most relevant result for the present work is Theorem \ref{P:prXannTannV}, which enables us to decide whether scalar Lie point symmetries are variational for a large class of theories, simply by considering the actions of the symmetries on the kinetic terms. This criterion is independent of the specific parameter conditions associated
with the often numerous potentials that realize the same symmetry algebra.
 The other two results, Proposition \ref{P:SimplifiedModel} and Corollary \ref{C:1d4dcorr}, were not strictly necessary for the present analysis, but they may substantially simplify the computations of variational, scalar Lie point symmetries of other, larger models of the type NHDM+KS. This is done by reducing the spacetime dimension, and hence the number of spacetime variables and gauge fields, while obtaining exactly the same results for scalar variational symmetries. Lie symmetry analysis detects only continuous (Lie) symmetries. Nevertheless, any discrete symmetries not detected by Lie's method may be identified by considering the automorphism groups of the Lie symmetry algebras computed through Lie symmetry analysis. In this way, the set of possible maximal, faithfully acting variational symmetry groups of a generic model can be determined. A systematic study of such automorphisms and their role in identifying the discrete symmetries of large quantum-physical models is left for future work.

\section*{Acknowledgments}
The author is grateful to Stylianos Dimas for providing the latest version of \texttt{SYM} \cite{dimas2005sym} and a \texttt{Mathematica} notebook demonstrating its use.

\appendix
\section{A proof of Proposition \ref{P:SimplifiedModel}}
\label{sec:ProofOfProposition}
The most general NHDM+KS Lagrangian in $d$ spacetime dimensions is given by
\eqref{E:L_NHDM+KS}.
 We can expand the Lagrangian \eqref{E:L_NHDM+KS} in purely scalar terms and terms linear and quadratic in the gauge fields, as follows:
\begin{align}\label{E:L_NHDM+KSgb}
	\mc{L}_d(x_0,\ldots,x_{d-1})&=\partial_\mu\Phi_n^\dag \partial^\mu \Phi_n+
	(\partial^\mu \Phi_n^\dag) G_\mu \Phi_n+\Phi_n^\dag G_\mu^\dag \partial^\mu \Phi_n+\Phi_n^\dag G_\mu^\dag G^\mu \Phi_n \nn \\
	&+\frac{1}{2}\partial_\mu s_m\partial^\mu s_m-V(\Phi,s)+T_\text{GB},
\end{align}
where 
\begin{align}
G^\mu=ig(\sigma_i/2)W^\mu_i + ig'YB^\mu	
\end{align}
represent the gauge bosons and where e.g.~$\mu \in \{0,1,2,3\}$ for $\mc{L}_4$ and $\mu=0$ for $\mc{L}_1$. Moreover, $T_\text{GB}$ is the sum of the kinetic terms of the gauge bosons
given by \eqref{E:TGB}, which is irrelevant for the proof of Proposition \ref{P:SimplifiedModel}, because we only consider scalar symmetries.
The relevant 1-prolongation of the scalar infinitesimal generator for the $d$-dimensional case is given by
\begin{align}\label{E:1prolongScalars}
	\on{pr}^{(1)} X_d&=\sum_{\mu=0}^{d-1} \sum_{i,j=1}^{4N+K}(\eta_i\partial_{\phi_i}+
	\frac{\partial \phi_j}{\partial{x_\mu}}\frac{\partial \eta_i}{\partial\phi_j}\frac{\partial}{\partial (\partial \phi_i/\partial x_\mu)}),
\end{align}
where we identify
\begin{align}\label{E:phisvsss}
	\phi_{4N+k}\equiv s_k,\quad \text{for} \; k \geq 1. 
\end{align}
The 1-prolongation \eqref{E:1prolongScalars} is of the given form since we only consider symmetries not involving $x_\mu$, i.e. 
$\xi^{\mu}=0$ and $\eta^j_{x_\mu}=0$ for all $0\leq \mu \leq d-1$ and $1\leq j\leq 4N+K$, since $\eta^j=\eta^j(\phi_1,\ldots,\phi_{4N},s_1,\ldots,s_K)$. Because we do not consider symmetries transforming the gauge bosons, $\eta$'s corresponding to the gauge fields are zero as well. 
Moreover, if $K=0$ in $\mc{L}_d(x_0,\ldots,x_{d-1})$, we have a pure NHDM, and if $N=1$ we have the SM augmented with $K$ real scalar gauge singlets.
Then, the scalar, variational symmetries of a theory given by a Lagrangian $\mc{L}_d$ are independent of $d$: 
\MainProp*
\begin{proof} 
First we consider strict variational symmetries, that is, symmetries with
$\on{pr}^{(1)}X_d\mc{L}_d=0$ (no sum over $d$), cf.~\eqref{E:condStrictVarSym} with $\xi=0$. 

 We denote 
$\partial_{\phi_i} \Phi_j\equiv  \Phi_{j,i}$, and similar for Hermitian conjugated fields. Then
e.g.~$\Phi_{1,4}^\dag =(0,-i)$. The effect of $\on{pr}^{(1)}X_d$ of \eqref{E:1prolongScalars} on $\mc{L}_d$, given by \eqref{E:L_NHDM+KSgb}, can be written as
\begin{align}\label{E:pr1vL1d4d}
	\on{pr}^{(1)}X_{d}\mc{L}_{d} &= 
	\eta_i (\partial^\mu \Phi_n^\dag) G_\mu \Phi_{n,i}+\eta_i \Phi_{n,i}^\dag G_\mu^\dag \partial^\mu \Phi_n+ \eta_i\Phi_{n,i}^\dag G_\mu^\dag G^\mu \Phi_n + \eta_i\Phi_{n}^\dag G_\mu^\dag G^\mu \Phi_{n,i}\nn \\
	&-\eta_i \frac{\partial V(\Phi,s)}{\partial \phi_i}
	+\frac{\partial \phi_j}{\partial{x_\mu}}\frac{\partial \eta_i}{\partial\phi_j}\frac{\partial \phi_i}{\partial x^\mu}
	+\frac{\partial \phi_j}{\partial{x_\mu}}\frac{\partial \eta_i}{\partial\phi_j}(\Phi_{n,i}^\dag G_\mu \Phi_n+\Phi_n^\dag G_\mu^\dag \Phi_{n,i}+\delta_{m,(i-4N)}\partial_\mu s_{m}) \nn \\
	&=0,
\end{align}
where $0\leq \mu \leq d-1$ and $1\leq K$. Assume that $\on{pr}^{(1)}X_1\mc{L}_1=0$ holds, we would like to show that then $\on{pr}^{(1)}X_d\mc{L}_d=0$ holds, with the same $\eta$'s. This means that a scalar symmetry of the $d=1$ theory also is a scalar symmetry of the theory with arbitrary dimension $d$. Now, $\eta_i {\partial_{\phi_i} V(\Phi)}{}=0$,
because all other terms are proportional to gauge bosons $G^\mu$ or
derivatives $\partial \phi_j /\partial x_\mu$ (which includes terms $\partial_\mu s_m$), and hence cannot cancel 
$\eta_i {\partial_{\phi_i} V(\Phi,s)}{}$. This implies that
$\eta_i {\partial_{\phi_i} V(\Phi)}{}=0$ in $\on{pr}^{(1)}X_d\mc{L}_d$ as well, because it is the same expression when suppressing spacetime variables.
The other terms all depend on $\mu$, and since they cancel for $\mu=0$ in $\on{pr}^{(1)}X_1\mc{L}_1$, they also have to cancel for each $\mu$ in
$\on{pr}^{(1)}X_4\mc{L}_4$, because the structure of the terms is the same for all $\mu$.

Conversely, assume $\on{pr}^{(1)}X_d\mc{L}_d=0$, that is, we have a symmetry of the theory of some fixed but arbitrary dimension $d>1$. Then, \eqref{E:pr1vL1d4d} holds where $\mu$ this time is a sum from $0$ to $d-1$. We will show that then $\on{pr}^{(1)}X_1\mc{L}_1=0$ for the same $\eta$'s, which means that the $d=1$ theory has the same symmetry. 
Again, because all other terms are proportional to gauge bosons or derivatives, $\eta_i {\partial_{\phi_i} V(\Phi)}{}=0$, and the same is true in the $d=1$ theory. Moreover, the other terms have to cancel for each $\mu$ individually, because terms involving gauge bosons $G^\mu$ cannot be canceled by terms involving gauge bosons $G^\nu$ for $\nu\ne \mu$, because they are independent fields.
Moreover, the term $\frac{\partial \phi_j}{\partial{x_\mu}}\frac{\partial \eta_i}{\partial\phi_j}\frac{\partial \phi_i}{\partial x^\mu}$
for a specific, fixed $\mu$ cannot be canceled by a term of the same type with a fixed spacetime index $\nu\ne \mu$, because $\eta_i$ does not involve spacetime derivatives of fields. No other terms can cancel these terms either. Therefore, terms containing $\mu$
in \eqref{E:pr1vL1d4d} cancel for each choice of $\mu$ separately, and hence they will cancel for $\mu=0$ in the $d=1$ theory as well, and the symmetry for the $d>1$ theory is also a symmetry for the $d=1$ theory. Thus, the scalar, SVSs for $\mc{L}_d$ and $\mc{L}_1$ are
the same. Hence, they are the same for $\mc{L}_d$ for all $d\in \N{N}$.

We now turn to the case where $\on{pr}^{(1)}X_d\mc{L}_d=\on{Div} \beta \equiv d_\mu \beta^\mu$ for some local vector fields $\beta^\mu(x,\phi,\phi_{,\alpha},\phi_{,\alpha\beta})$, which are not identically zero. 
The first-order infinitesimal variation $\delta \mc{L}_d=\on{pr}^{(1)}X_d\mc{L}_d$ is a divergence if and only if the Euler operator $E$ applied to $\on{pr}^{(1)}X_d\mc{L}_d$ is zero (cf.~\eqref{E:EtotDiv}),
which means that
\begin{align}
	E_j(\on{pr}^{(1)}X_d\mc{L}_d)=(\frac{\partial}{\partial y^j}-\frac{d}{dx^\mu}\frac{\partial}{\partial y^j_{,\mu}})(\on{pr}^{(1)}X_d\mc{L}_d)=0,
\end{align}
for all $1\leq j \leq (4N+K+4d)$, where the dependent variables $y^j$
are the $4N+K$ Higgs field components and the $4d$ gauge field components, and where $d$ again is the number of independent variables (the dimension).
First, assume $E(\on{pr}^{(1)}X_1\mc{L}_1)=0$, i.e.~we have a divergence symmetry in the $d=1$ theory.
The term $\eta_i {\partial_{\phi_i} V(\Phi,s)}{}$
depends only on undifferentiated scalars; hence, 
if $E_j$ corresponds to a Higgs field $y^j=\phi_j$, then 
$E_j(\on{pr}^{(1)}X_1\mc{L}_1)=0$ implies $\partial_{\phi_j}(\eta_i \partial_{\phi_i} V(\Phi,s))=0$. This is because the term $\partial_{\phi_j}(\eta_i \partial_{\phi_i} V(\Phi,s))$ cannot be canceled by any other term in $E_j(\on{pr}^{(1)}X_1\mc{L}_1)$
because they are all proportional to either
terms of the type $G_\mu$, $\partial_\mu \phi_k$ or a $\partial_\mu\partial^\mu \phi_k$
(all with $\mu=0$), in contrast to the term $\partial_{\phi_j}(\eta_i \partial_{\phi_i} V(\Phi,s))$.  
Hence, $\partial_{\phi_j}(\eta_i \partial_{\phi_i} V(\Phi,s))=0$ will also be true in the $d>1$ theory, as these terms are the same for all $d$, suppressing spacetime variables. All the remaining terms of \eqref{E:pr1vL1d4d} have obtained two factors of the type
$G_\mu$ or $\partial_\mu \phi_k$ for $\mu=0$, and regardless of which $E_i$
we apply, at least one such factor survives in each non-annihilated term; therefore, all these terms are tagged with the index $\mu=0$. Because the $d>1$ dimensional theory 
has $d$ copies of these terms, where each copy is tagged with an index $\mu$, $0\leq \mu \leq d-1$, each of the $d$ copies has to vanish separately because they are structurally the same, only distinguished by the index $\mu$. Thus, $E(\on{pr}^{(1)}X_d\mc{L}_d)=0$.

Conversely, assume $E(\on{pr}^{(1)}X_d\mc{L}_d)=0$, that is, we have a divergence symmetry in $d>1$ spacetime dimensions. By the same reason as before the terms 
$\partial_{\phi_j}(\eta_i \partial_{\phi_i} V(\Phi,s))=0$ in the $d>1$ theory, and hence the same is true in the $d=1$ theory, since those terms are the same for all choices of $d$. The remaining terms of $E(\on{pr}^{(1)}X_d\mc{L}_d)$
can be grouped into $d$ distinct, non-overlapping categories   
determined by the value of $\mu$ in factors of the types $G_\mu$,
$\partial_\mu \phi_k$ or $\partial_{\check{\mu}} \partial^{\check{\mu}} \phi_k$. 
The terms in each category must be canceled separately, since 
all terms in one category are linearly independent of any other term in another category.\footnote{\label{fnote:offShell} Note that the conditions for variational symmetries, including divergence symmetries, do not presuppose that the Euler-Lagrange equations of the Lagrangian hold. Hence, there is no relation between
$\partial_{\check{\mu}} \partial^{\check{\mu}} \phi_k$ and $\partial_{\check{\nu}} \partial^{\check{\nu}} \phi_k$ for $\check{\mu}\ne \check{\nu}$.} Hence the terms in the category $\mu=0$ cancel, and hence $E(\on{pr}^{(1)}X_1\mc{L}_1)=0$. Therefore, the scalar symmetries of $\mc{L}_d$ are the same for all $d$, because $\mc{L}_d$ has the same scalar symmetries as $\mc{L}_1$ for all $d$. Finally, we note that the arguments hold equally well if $K=0$, that is, we have a pure NHDM.
\end{proof}
We do not expect Proposition \ref{P:SimplifiedModel} to 
hold in general for non-variational symmetries, since the kinetic terms of the gauge bosons
consist of cross terms of different spacetime derivatives, and these terms will now matter, because a symmetry of the Euler-Lagrange equations
here implies that 
\begin{align}
	\on{pr}^{(2)}X_d (E(\mc{L}_d))=0 \quad \text{when}\quad E(\mc{L}_d)=0, 
\end{align}
and the kinetic terms of the gauge bosons will be involved in the "on-shell" condition $E(\mc{L}_d)=0$, while this was not the case for the "off-shell" (i.e., variational), scalar symmetries in Proposition \ref{P:SimplifiedModel}.
In contrast, the condition for a variational symmetry here is
\begin{align}
	E(\on{pr}^{(1)}X_d(\mc{L}_d))=0,
\end{align}
 with no assumption on the validity of the Euler-Lagrange equations.
Moreover, for $d=1$ the kinetic terms of the gauge bosons are zero because of the antisymmetry of the spacetime indices, whereas they are non-zero for $d>1$, which induces an asymmetry between the $d=1$ and $d>1$ cases in the on-shell condition $E(\mc{L}_d)=0$ for the two cases. Finally, the Euler-Lagrange equations imply dependencies between derivatives such as $\partial_{\check{\mu}} \partial^{\check{\mu}} \phi_k$ and $\partial_{\check{\nu}} \partial^{\check{\nu}} \phi_k$ for $\check{\mu}\ne \check{\nu}$, in contrast to the off-shell situation in Proposition \ref{P:SimplifiedModel} (cf.~footnote \ref{fnote:offShell}).

\section{Symmetries and changes of Higgs bases}
\label{sec:ChangesOfVariables}
To what extent does a change of Higgs basis change the symmetry of an equation or a Lagrangian $\mc{L}$, and how is it still detectable by Lie's method?
 As an example, consider a linear transformation of the scalars given by a matrix $S$, which is a strict variational symmetry if
\begin{align}
	\mc{L}(\phi)=\mc{L}(S\phi).
\end{align}
If we then consider a Higgs basis shift ${\phi}'=U\phi$, this possibly obfuscates a symmetry $S$, manifest in the unprimed basis $\phi$. The symmetry will still be present as an infinitesimal symmetry in the primed basis, although it may not be manifest (i.e.~it is hidden).  
For if 
\begin{align}
\mc{L}(\phi)=\mc{L}(U^\dag \phi')\equiv \mc{L}'(\phi'),	
\end{align}
then
$S'=USU^\dag $ is a symmetry of the transformed Lagrangian $\mc{L}'$ 
\begin{align}
	\mc{L}'(\phi')=\mc{L}'(S' \phi'),
\end{align}
if $S$ is a symmetry of the original Lagrangian, 
 since 
\begin{align}
	\mc{L}'(S' \phi')=\mc{L}(U^\dag US U^\dag \phi')=\mc{L}( {\phi})=\mc{L}'(\phi').
\end{align}
If $S$ is an infinitesimal symmetry, $S'$ is of course infinitesimal as well: 
For if $S=I+\epsilon X$ with $\epsilon$ an infinitesimal parameter, then $S'=I+\epsilon U X U^\dag$. Then the one-parameter group $\exp(t U X U^\dag)$, for $t\in \N{R}$, is continuous and connected to the identity, and is hence a Lie symmetry that will be detected by Lie symmetry analysis, albeit not a manifest symmetry.

\bibliographystyle{JHEP}

\bibliography{ref}

\end{document}